\documentclass{article}
\usepackage{arxiv}

\usepackage[utf8]{inputenc} 
\usepackage[T1]{fontenc}    
\usepackage{hyperref}       
\usepackage{url}            
\usepackage{booktabs}       
\usepackage{amsfonts}       
\usepackage{nicefrac}       
\usepackage{microtype}      
\usepackage{lipsum}

\usepackage[small]{caption}
\usepackage{algorithm}
\usepackage[noend]{algpseudocode}
\usepackage{amssymb}
\usepackage[round]{natbib}
\usepackage{mathtools}
\usepackage{amsmath}
\usepackage{tikz}
\usetikzlibrary{calc}
\usepackage{todonotes}
\usepackage{dsfont}
\usepackage{enumitem}
\usepackage{tikz}
\usepackage{amsthm}
\usepackage{thm-restate}
\usepackage{booktabs}
\usepackage{multirow}
\usepackage{newfloat}
\usepackage{wrapfig}
\usepackage{dsfont}
\usepackage{comment}
\usepackage{multirow}

\newtheorem{theorem}{Theorem}
\newtheorem{lemma}{Lemma}

\newtheorem{definition}{Definition}
\newtheorem{observation}{Observation}
\newtheorem{assumption}{Assumption}
\DeclareMathOperator*{\argmax}{argmax}

\title{Bayesian Agency: Linear versus Tractable Contracts}

\author{
	Matteo Castiglioni\\
	Politecnico di Milano\\
	\texttt{matteo.castiglioni@polimi.it}
	\And
	Alberto Marchesi\\
	Politecnico di Milano\\
	\texttt{alberto.marchesi@polimi.it}
	\And
	Nicola Gatti\\
	Politecnico di Milano\\
	\texttt{nicola.gatti@polimi.it}
}

\begin{document}


\maketitle

\begin{abstract}
	We study \emph{principal-agent} problems in which a principal commits to an outcome-dependent payment scheme (a.k.a. \emph{contract}) so as to induce an agent to take a costly, unobservable action.
	We relax the assumption that the principal perfectly knows the agent by considering a \emph{Bayesian} setting where the agent's type is unknown and randomly selected according to a given probability distribution, which is known to the principal.
	Each agent's type is characterized by her own action costs and action-outcome distributions.
	In the literature on non-Bayesian principal-agent problems, considerable attention has been devoted to \emph{linear contracts}, which are simple, pure-commission payment schemes that still provide nice approximation guarantees with respect to principal-optimal (possibly non-linear) contracts.
	While in non-Bayesian settings an optimal contract can be computed efficiently, this is no longer the case for our Bayesian principal-agent problems.
	This further motivates our focus on {linear contracts}, which can be optimized efficiently given their single-parameter nature.
	Our goal is to analyze the properties of linear contracts in Bayesian settings, in terms of approximation guarantees with respect to optimal contracts and general \emph{tractable contracts} (\emph{i.e.}, efficiently-computable ones).
	
	First, we study the approximation guarantees of linear contracts with respect to optimal ones, showing that the former suffer from a multiplicative loss that grows linearly in the number of agent's types.
	Nevertheless, we prove that linear contracts can still provide a constant multiplicative approximation $\rho$ of the optimal principal's expected utility, though at the expense of an exponentially-small additive loss $2^{-\Omega(\rho)}$.
	Then, we switch to tractable contracts, showing that, surprisingly, linear contracts perform well among them.
	In particular, we prove that it is \textsf{NP}-hard to design a contract providing a multiplicative loss sublinear in the number of agent's types, while the same holds for contracts that provide a constant multiplicative approximation $\rho$ at the expense of an additive loss $2^{-\omega(\rho)}$.
	We conclude by showing that, in Bayesian principal-agent problems, an optimal contract can be computed efficiently if we fix either the number of agent's types or the number of outcomes.
\end{abstract}


\section{Introduction}\label{sec:introduction}

\emph{Principal-agent} problems are ubiquitous in real-world economies.
These problems model interactions between two parties, a \emph{principal} and an \emph{agent}, where the latter chooses an action that determines some externalities on the former.
We focus on \emph{hidden-action} models, where the principal cannot observe the action taken by the agent, but only a stochastic outcome that is probabilistically determined as a result of the agent's action.
Each action is associated with a corresponding cost for the agent, while the principal receives a reward for the resulting outcome.
As a result, the principal's objective is to incentive an agent's action that leads to favorable outcomes.
This is achieved by committing to an outcome-dependent payment scheme, usually called \emph{contract}.

Principal-agent problems are pervasive in classical economic scenarios.
A well-known textbook example of principal-agent problem is that of a salesperson (agent) working for a company (principal).
The former has to decide on the level of effort she wants to put in selling products for the company.
Naturally, the company cannot observe the chosen level of effort (action), but it is only aware of the number of products sold.
Assuming that this figure is correlated with the level of effort selected by the salesperson, the company can incentivize an high level of effort by paying a commission to the salesperson based on the actual number of sales.

Interactions involving a principal and an agent play a crucial role also in modern economies centered around digital means.
In spite of this, principal-agent problems received far less attention from the economics and computation community than auctions and, more in general, mechanism design problems (more details on related computational works appear later in this section).
Remarkably, principal-agent models may have potential applications in various real-world settings, such as, \emph{e.g.}, crowdsourcing platforms~\citep{ho2016adaptive}, blockchain-based smart contracts~\citep{cong2019blockchain}, and healthcare~\citep{bastani2016analysis}.

In this paper, we study a generalization of the classical hidden-action principal-agent problem.
In particular, we relax the assumption that the principal perfectly knows the agent by considering a \emph{Bayesian} setting in which the agent's type is unknown and randomly selected according to a given probability distribution, which is known to the principal.
Each agent's type is characterized by her own action costs and action-outcome distributions.
In the salesperson example, types may correspond to different skill profiles for the salesperson, \emph{e.g.}, a clever worker can achieve better sales results than a non-clever one by putting the same level of effort in her work.

In the literature on principal-agent problems, considerable attention has been devoted to \emph{linear contracts} (see, \emph{e.g.},~\cite{carroll2015robustness,carroll2019robustness,dutting2019simple}), which are pure-commission payment schemes that pay the agent a given fraction of the principal's reward associated with the obtained outcome.
These contracts enjoy some nice properties.
In particular, they are simple to understand---given their single-parameter nature---and, in non-Bayesian settings, they still provide good approximation guarantees with respect to a principal-optimal (possibly non-linear) contract~\citep{dutting2019simple}.
While in non-Bayesian principal-agent problems an optimal contract can be computed efficiently by using a linear program, this is no longer the case in our Bayesian setting.
This further motivates our focus on {linear contracts}, which can be optimized efficiently given their single-parameter nature.

\subsection{Original Contributions}

The main goal of our work is to analyze the properties of linear contracts in Bayesian principal-agent settings, in order to understand their approximation guarantees with respect to optimal contracts and \emph{tractable} ones, with the latter being defined as those that can be computed efficiently (\emph{i.e.}, in polynomial time).
In particular, we look at approximations of the principal's expected utility. 
Notice that, while optimal contracts are a natural benchmark in any principal-agent problem, the comparison with tractable contracts becomes relevant and fundamental in our Bayesian model, where an optimal contract cannot be computed efficiently, and, thus, the most natural benchmark is the family of all contracts that can be computed in polynomial time. 
%

After introducing all the required preliminary concepts in Section~\ref{sec:preliminaries}, we start our analysis by studying, in Section~\ref{sec:linear}, the approximation guarantees of linear contracts with respect to optimal ones in Bayesian principal-agent problems.
We show that, from a purely-multiplicative approximation perspective, linear contracts suffer from a loss with respect to an optimal contract that grows linearly in the number of agent's types.
This happens in degenerate instances in which the principal's rewards are exponentially small in the number of agent's types, thus suggesting that there is hope linear contracts could obtain a constant multiplicative approximation, at the expense of an exponentially-small additive loss.
This motivates the introduction of $\big( \rho, g(\rho) \big)$\emph{-bi-approximate} contracts, which are those providing the principal with an expected utility at least $\frac{OPT}{\rho} - g(\rho)$, where $OPT$ is the principal's expected utility in an optimal contract.
Our main result is that linear contracts give a $\big( \rho, 2^{-\Omega(\rho)}\big)$\emph{-bi-approximation} of an optimal contract, \emph{i.e.}, they guarantee a constant multiplicative approximation $\rho$ of the optimal principal's expected utility, at the expense of an exponentially-small additive loss $2^{-\Omega(\rho)}$.
We complement this result by showing that no linear contract can provide a $\big(  \rho ,2^{-\omega \left( \rho \right)} \big)$-approximation of an optimal one, even in non-Bayesian settings.
This implies that, using linear contracts, we can only obtain bi-approximations whose additive losses decrease at most exponentially in the multiplicative factor $\rho$.
Notice that our bi-approximation results also hold for the basic non-Bayesian case, complementing known approximation results of linear contracts in such setting~\citep{dutting2019simple} (see the related works for more details).

Then, in Section~\ref{sec:hardness}, we focus on the performances of linear contracts with respect to tractable ones in Bayesian settings, showing that, surprisingly, they perform well.
In particular, we show that there is no tractable contract providing a constant multiplicative loss with respect to an optimal one.
Formally, we prove that it is $\mathsf{NP}$-hard to design a contract with a multiplicative loss sublinear in the number of agent's types.
Then, we study the approximation guarantees of tractable contracts in terms of bi-approximations.
We prove that it is $\mathsf{NP}$-hard to design a contract that provides a $\big( \rho,2^{-\omega(\rho)} \big)$-bi-approximation of an optimal one, thus matching the lower bound of linear contracts.

We conclude with Section~\ref{sec:tractable}, where we show that there are some special cases of our Bayesian principal-agent problem in which an optimal contract can be computed in polynomial time.
In particular, this happens if we fix either the number of agent's types or the number of outcomes.

\subsection{Related Works}

Hidden-action principal-agent problems have received considerable attention in the economic literature, where they usually fall under the umbrella of a broader subject called \emph{contract theory}, which is a fundamental pillar of microeconomic theory~\citep{shavell1979risk,grossman1983analysis,rogerson1985repeated,holmstrom1991multitask} (see the books by~\citet{mas1995microeconomic},~\citet{bolton2005contract},~and~\citet{laffont2009theory} for a detailed treatment of the subject).

The first computational studies on principal-agent problems appeared only recently.
Among them, it is worth discussing in detail that of~\citet{dutting2019simple}, which is perhaps the most related to ours.
\citet{dutting2019simple} study non-Bayesian principal-agent problems (\emph{i.e.}, a special case of our setting having only one agent's type), with a focus on linear contracts.
In particular, they show that linear contracts provide a constant multiplicative approximation of the principal's expected utility in an optimal contract, except in degenerate instances having the following three properties \emph{simultaneously}: there are many agent's actions, there is a big spread of rewards, and there is a big spread of costs.
Moreover, the results of~\citet{dutting2019simple} are tight.
In our work, we extend this comparison between linear and optimal contracts to our Bayesian settings.
However, apart from that, our work considerably departs from~\citep{dutting2019simple}, since our main focus is on understanding the performances of linear contracts with respect to tractable ones.
Notice that this is \emph{not} a concern for~\citet{dutting2019simple}, since, differently from the Bayesian setting, an optimal contract can be computed efficiently in classical (non-Bayesian) principal-agent problems.

There is a number of other computational works that study extensions of classical hidden-action principal-agent problems exhibiting some sort of combinatorial structure.
For instance, the work of~\citet{babaioff2006combinatorial} studies a model with multiple agents (see also its extended version~\citep{babaioff2012combinatorial} and its follow-ups~\citep{babaioff2009free,babaioff2010mixed}).
Its focus is on how complex combinations of agents' actions influence the resulting outcome in presence of inter-agent externalities, while in our model there is only one agent that can be of different types, and, thus, no externalities among agent's types are involved.
Moreover,~\citet{babaioff2006combinatorial} study settings in which each agent has only two actions, while in our model each agent's type can have an arbitrary number of actions.
Recently,~\citet{dutting2020complexity} study another principal-agent problem whose underlying structure is combinatorial, as a result of defining the outcome space implicitly through a suitably-defined succinct representation.

Other computational works on principal-agent problems worth citing are~\citep{babaioff2014contract}, which introduces a notion of contract complexity based on the number of different payments specified by the contract, and~\citep{ho2016adaptive}, which develops a dynamic model where, in each round, the principal determines a contract and an agent chooses an action, resulting in a reward for the principal.
These works considerably depart form ours, as they study rather different models.
The first one considers an $n$-player normal-form framework in which actions are \emph{not} hidden.
The second work uses multi-armed bandit techniques, and, thus, the goal is to minimize the principal’s regret over time.

In conclusion, we also point out that considerable attention (especially in the economic literature) has been devoted to the study of some \emph{robustness} properties of linear contracts in classical principal-agent problems~\citep{carroll2015robustness,carroll2019robustness}.
This perspective has also been taken by~\citet{dutting2019simple} using a more computationally-oriented point of view.

\paragraph{Note on Concurrent Work by~\citet{guruganesh2020contracts}}
The work by~\citet{guruganesh2020contracts}, which has been developed independently and concurrently with ours, studies the same Bayesian principal-agent problem that we address in this paper.
\citet{guruganesh2020contracts} characterize worst-case multiplicative approximation guarantees of linear contracts, comparing them with some benchmarks (including optimal contracts).
%
Among the results they provide, the closest to ours are discussed in the following.
First, they show a tight approximation guarantee for linear contracts, which is linear in the number of agent's actions and logarithmic in the number of agent's types when all agent's types share the same costs, while, if they may have different costs, it is linear in the number of types and actions.
This result is similar to our result in Section~\ref{sec:linear_vs_optimal}, where we only consider the dependency on the number of types.
Second, they show the hardness of computing a single contract or a menu of contracts approximating the optimal principal's expected utility ip to within a given constant multiplicative factor.
In Section~\ref{sec:hardness}, we show a stronger result.
In particular, we prove the hardness of computing a single contract with a multiplicative loss sublinear in the number of types.
Finally, they show that an optimal contract can be computed efficiently if we fix either the number of agent’s types or the number of outcomes.
This is equivalent to our results in Section~\ref{sec:tractable}.
In conclusion, even though~\citet{guruganesh2020contracts} study the same principal-agent problem, they focus on the approximation guarantees of linear contracts with respect to optimal ones and other possible benchmarks, while our main focus is their relation with efficiently computable contracts.

\section{Preliminaries}\label{sec:preliminaries}

%
In this section, we introduce all the elements we need in the rest of this work.
Section~\ref{sec:preliminaries_problem} formally defines the problem we study, Section~\ref{sec:preliminaries_contracts} describes its solutions (contracts), while Section~\ref{sec:preliminaries_apx} defines which kind of approximation guarantees we look for in contracts.

\subsection{The Bayesian Principal-Agent Problem}\label{sec:preliminaries_problem}

An instance of the \emph{Bayesian principal-agent problem} is characterized by a tuple $(\Theta,A,\Omega)$, where: $\Theta$ is a finite set of $\ell \coloneqq |\Theta|$ agent's types; $A$ is a finite set of $n \coloneqq |A|$ actions available to the agent; and $\Omega$ is a finite set of $m \coloneqq |\Omega|$ possible outcomes.~\footnote{For the simplicity of exposition, we assume that all the agent's types share the same action set. All the results continue to hold even if each agent's type $\theta \in \Theta$ has her own action set $A_\theta$.}
The agent's type is drawn according to a fixed probability distribution known to the principal. 
We let $\mu \in \Delta_{\Theta}$ be such a distribution, with $\mu_\theta$ denoting the probability of type $\theta \in \Theta$ being selected.~\footnote{Given a finite set $X$, we denote with $\Delta_X$ the set of all the probability distributions defined over $X$.}
For each type $\theta \in \Theta$, we introduce $F_{\theta, a} \in \Delta_\Omega$ to denote the probability distribution over outcomes $ \Omega$ when an agent of type $\theta$ selects action $a \in A$, while $c_{\theta, a} \in [0,1]$ is the agent's cost for that action.~\footnote{For the ease of presentation, we assume that rewards and costs are in $[0,1]$. Notice that all the results in this work can be easily generalized to the case of an arbitrary range of positive numbers, by applying a suitable normalization.}
We let $F_{\theta, a, \omega} $ be the probability that $F_{\theta, a}$ assigns to $\omega \in \Omega$, so that $\sum_{\omega \in \Omega} F_{\theta, a, \omega}  =1$.
Each outcome $\omega \in \Omega$ is characterized by a reward $r_\omega \in [0,1]$ for the principal.
As a result, when an agent of type $\theta \in \Theta$ selects an action $a \in A$, then the principal achieves an expected reward $R_{\theta,a}$, which is defined as $R_{\theta,a} \coloneqq \sum_{\omega \in \Omega} F_{\theta, a, \omega} \, r_\omega$.
As in classical (non-Bayesian) principal-agent problems, the principal's objective is to commit to a contract that maximizes her expected utility, as we formally describe in the following.

\subsection{Contracts}\label{sec:preliminaries_contracts}

A \emph{contract} is specified by payments from the principal to the agent, which are contingent on the actual outcome achieved with the agent's action.
We let $p_\omega \ge 0$ be the payment associated to outcome $\omega \in\Omega$.
The assumption that payments are non-negative (\emph{i.e.}, they can only be from the principal to the agent, and \emph{not} the other way around) is common in contract theory, where it is known as \emph{limited liability}~\citep{carroll2015robustness}.
When an agent of type $\theta \in \Theta$ selects an action $a \in A$, then the expected payment to the agent is $P_{\theta,a} \coloneqq \sum_{\omega \in \Omega} F_{\theta, a, \omega} \, p_\omega$, while her utility is $P_{\theta,a} - c_{\theta, a}$.
On the other hand, the principal's expected utility in that case is $R_{\theta,a} -P_{\theta,a}$.

Given a contract, an agent of type $\theta \in \Theta$ selects an action such that:
\begin{enumerate}
	\item it is \emph{incentive compatible} (IC), \emph{i.e.}, it maximizes her expected utility among actions in $A$;
	\item it is \emph{individually rational} (IR), \emph{i.e.}, it has non-negative expected utility (if there is no IR action, then the agent of type $\theta$ abstains from playing so as to maintain the \emph{status quo}).
\end{enumerate}

For the ease of presentation, we adopt the following w.l.o.g. common assumption~\citep{dutting2019simple}, which guarantees that IR is always enforced and, thus, it allows us to focus on IC only.
\begin{assumption}\label{ass:ir}
	There exists an action $a \in A$ such that $c_{\theta, a} = 0$ for all $\theta \in \Theta$.
\end{assumption} 
The assumption ensures that each agent's type has always an action providing her with a non-negative utility, thus ensuring IR of any IC action.

We say that a contract \emph{implements} an action $a^\ast \in A$ for an agent of type $\theta \in \Theta$ if the agent chooses that action.~\footnote{As it is common in the literature~\citep{dutting2020complexity}, we assume that the agent breaks ties in favor of the principal, \emph{i.e.}, whenever there is more than one IC action, she selects the one maximizing the principal's expected utility.}
Finally, given a contract, by letting $a^\ast(\theta) \in A$ be the action implemented by such contract for an agent of type $\theta \in \Theta$, we can define the overall principal's expected utility as $ \sum_{\theta \in \Theta} \mu_\theta \left( R_{\theta,a^\ast (\theta)} -P_{\theta,a^\ast (\theta)} \right)$, which accounts for type probabilities.
%
%
A special class of simple contracts that is commonly studied in the literature is that of \emph{linear contracts}, which give payments equal to some fixed fraction of the outcome rewards~\citep{dutting2019simple}.
Thus, these contracts are completely characterized by a single parameter $\alpha \in [0,1]$, with their payments being defined as $p_\omega = \alpha \, r_\omega$ for every outcome $\omega \in \Omega$.
We refer the reader to the work by~\citet{dutting2019simple} for more details on linear contracts in non-Bayesian principal-agent problems, including their geometric interpretation.

\subsection{Approximation Guarantees of Contracts}\label{sec:preliminaries_apx}

The goal of the principal is to design an \emph{optimal} contract, which is one maximizing her overall expected utility.
In the following, we denote with $OPT$ the principal's overall expected utility in an optimal contract.
As we show later in this work, computing an optimal contract in our setting is computationally intractable (with the exception of some special cases).
Thus, we look at suboptimal contracts providing some guaranteed approximation of the principal's optimal utility $OPT$.

Given a contract, we say that its {multiplicative loss} with respect to an optimal contract is $\rho \geq 1$ if it provides the principal with an overall expected utility of $\frac{OPT}{\rho}$.
Equivalently, we sometimes say that the contract provides a multiplicative approximation $\rho$ of an optimal one.
%

We also study approximation guarantees of contracts by considering both additive and multiplicative approximations at the same time.
Formally, given a multiplicative approximation $\rho \geq 1$, we say that a contract provides a $\left( \rho, g(\rho) \right)$\emph{-bi-approximation} of an optimal one if it results in an overall principal's expected utility greater than or equal to $\frac{OPT}{\rho} - g(\rho)$, where $g(\rho)$ denotes a (positive) additive loss depending on the parameter $\rho$.
Intuitively, bi-approximations allow us to analyze the performances of contracts by carefully managing the trade off between a desired (constant) multiplicative approximation factor and an additional (small) additive loss.
We are interested in $\left( \rho, g(\rho) \right)$-bi-approximations such that the term $g(\rho)$ quickly approaches zero as $\rho$ increases.
In particular, later in this work, we focus on bi-approximations whose $g(\rho)$ terms decrease exponentially in $\rho$, so that the additive loss becomes quickly negligible.

%
%

\section{Linear versus Optimal Contracts: The Bayesian Setting}\label{sec:linear}

We start by analyzing the performances of linear contracts with respect to optimal (possibly non-linear) ones.
Section~\ref{sec:linear_vs_optimal} studies how linear contracts perform in terms of multiplicative loss, while Section~\ref{sec:linear_bi_apx} provides our main results on the bi-approximation guarantees of linear contracts.

\subsection{Multiplicative Approximations}\label{sec:linear_vs_optimal}

We prove that, in Bayesian principal-agent problems, linear contracts do \emph{not} perform well when compared to optimal ones in terms of their multiplicative loss.
%
Formally, in the following Theorem~\ref{thm:linear_lower_bound}, we construct particular instances showing that, in the worst case, the multiplicative loss of any linear contract with respect to an optimal one increases at least linearly in the number of agent's types $\ell$.
%
%
We remark that, in the instances used to prove the theorem, the agent has only two actions available (notice that $A = \{ a_1, a_2 \}$ in the proof).
This strengthens already-known results.
Indeed,~\citet{dutting2019simple} prove that, in the special case of non-Bayesian principal-agent problems, linear contracts are arbitrarily worse than optimal ones for a growing number of agent's actions, as their worst-case multiplicative loss is equal to $n$.
Our result shows that, in Bayesian settings with many agent's types, the multiplicative loss of linear contracts can be arbitrarily bad even in the basic case in which the agent has only two actions. 
%

\begin{theorem}\label{thm:linear_lower_bound}
	In Bayesian principal-agent problems, the worst-case multiplicative loss of any linear contract with respect to an optimal one is $\Omega(\ell)$, where $\ell$ is the number of agent's types.
	%
	%
\end{theorem}

\begin{proof}
	For any $\ell \in \mathbb{N}$, let us consider a principal-agent setting $(\Theta, A, \Omega)$ with outcome set $\Omega = \{ \omega_j \}_{j \in [m]}$, where we let $m = \ell +1$.
	We define $r_{\omega_j} = 2^{-j}$ for $\omega_j \in \Omega \setminus \{ \omega_m \}$, while $r_{\omega_m} = 0$.
	The set of agent's types is $\Theta = \{ \theta_k \}_{k \in [\ell]}$, with $\mu \in \Delta_\Theta$ being defined so that $\mu_{\theta_k} = \frac{1}{N} 2^{2(k - \ell)}$ for all $\theta_k \in \Theta$, where $N \coloneqq \sum_{k \in [\ell]} 2^{2(k - \ell)}$ is a suitably defined normalization constant.
	Each agent of type $\theta_k \in \Theta$ has two actions available, namely $A = \{ a_1, a_2 \}$, with probability distributions defined so that $F_{\theta_k,a_1,\omega_k} = 1$ and $F_{\theta_k,a_2,\omega_m} = 1$.
	Intuitively, action $a_1$ of type $\theta_k$ deterministically results in outcome $\omega_k$ (with reward $r_{\omega_k} = 2^{-k}$), while action $a_2$ leads to outcome $r_{\omega_m}$ no matter the agent's type (with reward $r_{\omega_m} = 0$).
	Moreover, the action costs for type $\theta_k \in \Theta$ are $c_{\theta_k,a_1} = 2^{-k} \left(  1- 2^{-k}  \right)$ and $c_{\theta_k,a_2}= 0$.
	The optimal (non-linear) contract sets the payments as follows: $p_{\omega_j} = 2^{-j} \left( 1-2^{-j} \right)$ for all $\omega_j \in \Omega \setminus \{ \omega_m \}$, while $p_{\omega_m} = 0$.
	This contract implements action $a_1$ for each agent's type $\theta_k \in \Theta$, as her utility by playing $a_1$ is $P_{\theta_k, a_1} - c_{\theta_k, a_1} = p_{\omega_k} - c_{\theta_k,a_1} = 0$ and $r_{\omega_k} > 0$, while the utility of $a_2$ is zero and $r_{\omega_m} = 0$ (as previously stated, we assume that ties are broken in favor of the principal).
	As a result, the contract provides the principal with an overall expected utility of:
	\[
		\sum_{\theta_k \in \Theta} \mu_{\theta_k} \left(  R_{\theta_k, a_1} - P_{\theta_k, a_1}  \right) = \frac{1}{N}\sum_{\theta_k \in \Theta} 2^{2(k - \ell)} \left[  2^{-k} - 2^{-k} \left( 1-2^{-k} \right)  \right] =\frac{1}{N}\sum_{\theta_k \in \Theta}   2^{-2\ell}= \frac{\ell \,  2^{-2\ell}}{N} .
	\]
	Now, let us consider a linear contract with parameter $\alpha \in [0,1]$.
	For an agent of type $\theta_k \in \Theta$, the contract implements action $a_1$ only if $P_{\theta_k, a_1} = p_{\omega_k} = \alpha\, r_{\omega_k} = \alpha\, 2^{-k} \geq  2^{-k} \left(  1- 2^{-k}  \right) = c_{\theta_k,a_1}$.
	It is easy to check that an agent of type $\theta_k$ is incentivized to play $a_1$ if and only if $k \leq - \log_2 (1 -\alpha)$.
	Let $k' \coloneqq \left\lfloor - \log_2 (1 -\alpha) \right\rfloor$ be the highest index among agent's types that are incentivized to play action $a_1$.
	%
	%
	%
	Then, the overall principal's expected utility is:
	\begin{align*}
		\sum_{\theta_k \in \Theta: k \leq k'} \mu_{\theta_k} (1 -\alpha) \, R_{\theta_k, a_1} & =\frac{1}{N} \sum_{\theta_k \in \Theta: k \leq k'}2^{2(k - \ell)} (1 -\alpha) \, 2^{-k} \leq \\
		& \leq \frac{1}{N} \sum_{\theta_k \in \Theta: k \leq k'}2^{2(k - \ell)} \, 2^{-k'} \, 2^{-k} = \frac{1}{N} 2^{-2 \ell} \, 2^{-k'} \sum_{\theta_k \in \Theta: k \leq k'} 2^{k} \leq \\
		 & \leq \frac{1}{N} 2^{-2 \ell} \,2^{-k'} \,2^{k' +1} = \frac{2 \, 2^{-2 \ell}}{N},
	\end{align*}
	where the first inequality follows from $1 -\alpha \leq 1 - 2^{-k} \leq 1 - 2^{-k'}$ (since the contract implements $a_1$ for type $\theta_k$ and it holds $k \leq k'$). 
	This concludes the proof.
\end{proof}

We remark that the approximation result in Theorem~\ref{thm:linear_lower_bound} is tight in many cases.
%
%
This is readily seen by leveraging the approximation results of~\citet{dutting2019simple}.
Let us recall that~\citet{dutting2019simple} show that, in non-Bayesian principal-agent problems, linear contracts provide a constant multiplicative approximation of optimal ones, except in settings where the following three conditions hold \emph{simultaneously}: there are many agent's actions, there is a big spread of expected rewards, and there is a big spread of costs.
Thus, whenever at least one of the conditions above does \emph{not} hold in a Bayesian principal-agent setting, we have a simple polynomial-time algorithm that returns a linear contract with multiplicative loss $O(\ell)$ (matching the lower bound in Theorem~\ref{thm:linear_lower_bound}).
This algorithm computes an approximate linear contract of~\citet{dutting2019simple} for each agent's type {singularly} and returns the one providing the highest overall principal's expected utility (after weighting them by the corresponding type probabilities).~\footnote{For each agent's type $\theta \in \Theta$, the algorithm computes the linear contract of Theorem~5.1~in~\citep{dutting2019simple} if the number of actions is small (that of Theorem~5.5, respectively Theorem~5.7,~in~\citep{dutting2019simple} if the spread of rewards, respectively costs, is small).}
%
%
Let us also notice that, even in pathological cases in which all the conditions above hold simultaneously, the result in Theorem~\ref{thm:linear_lower_bound} is still tight in the number of agent's types $\ell$, though the multiplicative loss of linear contracts could be arbitrarily bad in one or more of the other parameters (number of agent's actions, spread of rewards, and spread of costs).
{For instance, \citet{guruganesh2020contracts} show that the worst-case loss is linear in the number of actions.}

\subsection{Bi-Approximation Guarantees}\label{sec:linear_bi_apx}

%
The instances exploited in the proof of Theorem~\ref{thm:linear_lower_bound} suggest that the negative result holds only when the rewards (and, thus, the principal's expected utilities) are very small.
In particular, the rewards decrease exponentially in the approximation factor (the number of agent's types).
This suggests that linear contracts could provide nice approximation guarantees when looking at bi-approximations. 

Next, we prove that linear contracts achieve good bi-approximations of the optimal principal's expected utility: for any constant $\rho$, 
they provide a multiplicative approximation $\rho$ at the expense of an exponentially small additive loss $2^{-\Omega \left(  \rho \right) }$.
%
%
Let us remark that the additive loss decreases exponentially as $\rho$ increases, becoming quickly negligible.
For instance, given a constant multiplicative approximation factor $\rho = 50$, the resulting additive loss is $2^{-24}$, while Theorem~\ref{thm:linear_lower_bound} shows that linear contracts provide a (non-constant) approximation decreasing linearly in the number of agent's types $\ell$ if we only consider multiplicative factors.
%
%
%

We start by proving the result in the easier non-Bayesian setting, in which $\ell = 1$.~\footnote{When we refer to a non-Bayesian principal-agent instance, we adopt the same notational conventions, dropping any reference to the types.}

First, we introduce the following observation that is useful to prove the following Theorem~\ref{thm:mult_add_no_bayes}, as well as other results in the rest of this work.
%

\begin{observation}\label{obs:opt_bound}
	Given a non-Bayesian principal-agent instance, it holds $OPT \leq \max_{a \in A} \left\{  R_a - c_a \right\}$.
\end{observation}
\begin{proof}
	It is sufficient to notice that, by the IR property, the agent's expected payment covers the cost $c_a$ of the implemented action $a\in A$, and, thus, the principal's expected utility is always upper-bounded by $R_a - c_a$.
\end{proof}

\begin{theorem}\label{thm:mult_add_no_bayes}
	Given a non-Bayesian principal-agent instance, linear contracts provide a $\big( \rho, 2^{-\Omega \left( \rho \right)} \big)$-bi-approximation of an optimal contract.
	Moreover, for any $\rho \geq 2$, there is a linear contract $\alpha$ that provides a $\left( \rho, 2^{- d \rho + e } \right)$-bi-approximation of an optimal contract for two constants $d \in \mathbb{R}^+$ and $e \in \mathbb{R}$, where $\alpha = 1 - 2^{-i}$ for some $i = 1, \ldots, \left\lfloor \rho /2 \right\rfloor$.~\footnote{A bi-approximation as in Theorems~\ref{thm:mult_add_no_bayes}~and~\ref{thm:mult_add_bayes} for the cases in which $\rho \in [1,2)$ can be easily obtained by using the linear contract $\alpha=\frac{1}{2}$, which always provides an additive loss of $\frac{1}{2}$.}
\end{theorem}

\begin{proof}
	For the ease of presentation, given $\rho \ge 2$, we let $I \coloneqq \lfloor \rho / 2  \rfloor $, while $[I]$ is the set of integers from $1$ to $\lfloor \rho / 2  \rfloor$.
	Moreover, for each $i \in [I]$, we define $\alpha_i \coloneqq 1 - 2^{-i} $, while, letting $A = \{ a_i \}_{i \in [n]}$, we assume w.l.o.g. that the first $I$ actions of $A$ are those implemented by the linear contracts with parameters $\alpha_i$, so that $a_i \in A $ denotes the agent's action implemented by $\alpha_i$.
	For $i \in [I]: i > 1$, we also let $\alpha_{i-1, i} \in \left[  \alpha_{i-1}, \alpha_i \right]$ be the parameter identifying a linear contract such that the agent is indifferent between actions $a_{i-1}$ and $a_i$.
	Whenever $a_{i-1}$ and $a_i$ are the same action, then we can set w.l.o.g. $\alpha_{i-1, i} \coloneqq \alpha_i$.
	Instead, if $a_{i-1}$ and $a_i$ are different actions, it is easy to check that it must be $\alpha_{i-1, i} \coloneqq \frac{c_{a_{i-1}} - c_{a_i}}{R_{a_{i-1}} - R_{a_i}}$ (it cannot be the case that $R_{a_{i-1}} = R_{a_i}$, otherwise one between $a_{i-1}$ and $a_i$ would be weakly dominated by the other, and, thus, never implemented by a contract that breaks ties deterministically in favor of the principal).
	Finally, for convenience, we let $\alpha_{0,1} \coloneqq 0$.
	In the following, we show that at least one linear contract among those with parameters $\alpha_i$ for $i \in [I]$ provides the principal with expected utility at least $\frac{OPT}{\rho} - 2^{-\frac{\rho}{2}+1 }$.
	
	In the rest of the proof, we need the following observation due to~\citet{dutting2019simple}.
	\begin{observation}[Essentially Observation~6 in~\citep{dutting2019simple}]\label{obs:delta_welfare}
		Given $i \in [I] : i > 1$, it holds:
		\[
			\left(  R_{a_i} - c_{a_i} \right) - \left(  R_{a_{i-1}} - c_{a_{i-1}} \right) \leq 	\left(  1 - \alpha_{i-1, i} \right) R_{a_i}.
		\]
	\end{observation}
	Observation~\ref{obs:delta_welfare} allows us to prove the following lemma.
	\begin{lemma}\label{lem:telescoping}
		For every  $i' \in [I]$, it holds that:
		\[
			R_{a_{i'}} - c_{a_{i'}} \leq \sum_{i \in [I]: i \leq i'} \left(  1  -\alpha_{i-1,i} \right) R_{a_i} .
		\]
	\end{lemma}
	\begin{proof}
		The proof is by induction.
		For the base case $i' = 1$, we have $\left( 1 - \alpha_{0,1} \right) R_{a_1} = R_{a_1} \geq R_{a_1} - c_{a_1}$.
		Next, for every $i' \in [I]: i'\geq 2$, let us assume by induction that $R_{a_{i' - 1}} - c_{a_{i' -1}} \leq \sum_{i \in [I]: i \leq i' -1} \left(  1  -\alpha_{i-1,i} \right) R_{a_i} $.
		Then, by using Observation~\ref{obs:delta_welfare} and the inductive hypothesis, we get:
		\begin{align*}
			R_{a_{i'}} - c_{a_{i'}} & = R_{a_{i'}} - c_{a_{i'}} - \left(  R_{a_{i' -1}} - c_{a_{i' -1}} \right) + \left(  R_{a_{i' -1}} - c_{a_{i' -1}} \right) \leq \\
			& \leq \left(  1  -\alpha_{i'-1,i'} \right) R_{a_{i'}} + \sum_{i \in [I]: i \leq i' -1} \left(  1  -\alpha_{i-1,i} \right) R_{a_i} = \\
			& = \sum_{i \in [I]: i \leq i'} \left(  1  -\alpha_{i-1,i} \right) R_{a_i}.
		\end{align*}
		This concludes the proof of the lemma.
	\end{proof}
	
	Now, we can prove the following:
	\[
		\max_{i \in [I]} \left(  1  - \alpha_{i}  \right) R_{a_i} \geq \frac{1}{2} \max_{i \in [I]} \left(  1  - \alpha_{i-1, i}  \right) R_{a_i} \geq \frac{1}{2 I} \sum_{i \in [I]} \left( 1 - \alpha_{i-1, i}  \right) R_{a_i} \geq \frac{R_{a_I} - c_{a_I}}{2 I},
	\]
	where the first inequality holds since $1 - \alpha_{i} = 2^{-i} = \frac{1}{2} 2^{- \left(  i -1 \right)} = \frac{1}{2} \left( 1 - \alpha_{i - 1} \right) \geq \frac{1}{2} \left( 1 - \alpha_{i - 1, i} \right) $, the second one holds by definition of $\max$, while the last one by Lemma~\ref{lem:telescoping}.
	Finally, by letting $a^\ast \in \argmax_{a \in A} \left\{  R_a - c_a \right\}$, we have:
	\[
		R_{a_I} - c_{a_I} \geq \alpha_{I}  R_{a_I} - c_{a_I} \geq \alpha_{I}  R_{a^\star} - c_{a^\star} \geq \left( 1 - 2^{-I} \right) R_{a^\star} - c_{a^\star} \geq R_{a^\star} - c_{a^\star} - 2^{-I} \geq OPT - 2^{-I},
	\]
	where the second inequality holds since the linear contract with parameter $\alpha_{I}$ implements action $a_I$, the second-last inequality follows from $R_{a^\star} \in [0,1]$, while the last one holds by Observation~\ref{obs:opt_bound}. 
	%
	In conclusion,
	\[
		\max_{i \in [I]} \left(  1  - \alpha_{i} \right) R_{a_i} \geq \frac{R_{a_I} - c_{a_I}}{2 I} \geq \frac{OPT - 2^{-I}}{2 I} \geq   \frac{OPT}{\rho} - 2^{-\frac{\rho}{2} +1 }.
	\]
	This concludes the proof of the theorem. 
	%
\end{proof}

We can exploit a reasoning similar to that used in the proof of Theorem~\ref{thm:mult_add_no_bayes} to prove our main result for the general Bayesian setting.

\begin{theorem}\label{thm:mult_add_bayes}
	Given a Bayesian principal-agent instance, linear contracts provide a $\big( \rho, 2^{-\Omega \left( \rho \right)} \big)$-bi-approximation of an optimal contract.
	Moreover, for any $\rho \geq 2$, there is a linear contract $\alpha$ that provides a $\left( \rho, 2^{- d \rho + e } \right)$-bi-approximation of an optimal contract for two constants $d \in \mathbb{R}^+$ and $e \in \mathbb{R}$, where $\alpha = 1 - 2^{-i}$ for some $i = 1, \ldots, \left\lfloor \rho /2 \right\rfloor$.
\end{theorem}

\begin{proof}
	The proof follows the lines of that of Theorem~\ref{thm:mult_add_no_bayes}, where we let $I \coloneqq \lfloor \rho / 2  \rfloor$ and $\alpha_i \coloneqq 1 - 2 ^{-i}$ for $i \in [I]$.
	In this case, for every agent's type $\theta \in \Theta$, with a slight abuse of notation we define $a_i^\ast(\theta) \in A$ as the action implemented by the linear contract with parameter $\alpha_i$ for an agent of type $\theta$.
	Moreover, for every $i \in [I]$, we introduce the parameters $\alpha_{a_{i-1}^\ast(\theta) , a_i^\ast(\theta) }$ (with $\alpha_{a_{0}^\ast(\theta) , a_1^\ast(\theta) } \coloneqq 0$), which are the analogous of the parameters $\alpha_{i - 1, i}$ in the proof of Theorem~\ref{thm:mult_add_no_bayes}, for actions $a_{i-1}^\ast(\theta)$ and $a_{i}^\ast(\theta)$.
	Then, following steps similar to those in Theorem~\ref{thm:mult_add_no_bayes} (including an analogous of Lemma~\ref{lem:telescoping}), we can prove the following:
	\begin{align*}
		\max_{i \in [I]} \sum_{\theta \in \Theta} \mu_{\theta} \left( 1 - \alpha_i \right) R_{\theta, a_i^\ast(\theta)} & \geq  \frac{1}{2} \max_{i \in [I]} \sum_{\theta \in \Theta} \mu_{\theta} \left( 1 - \alpha_{a_{i-1}^\ast(\theta) , a_i^\ast(\theta) } \right) R_{\theta, a_i^\ast(\theta)} \geq \\
		& \geq  \frac{1}{2 I} \sum_{\theta \in \Theta} \mu_{\theta} \sum_{i \in [I]} \left( 1 - \alpha_{a_{i-1}^\ast(\theta) , a_i^\ast(\theta) } \right) R_{\theta, a_i^\ast(\theta)} \geq \\
		& \geq  \frac{1}{2 I} \sum_{\theta \in \Theta} \mu_{\theta} \sum_{i \in [I]} \left( R_{\theta, a_I^\ast(\theta)}  - c_{\theta,a_I^\ast(\theta) } \right) \geq \\
		& \geq 	\frac{1}{2 I} \sum_{\theta \in \Theta} \mu_{\theta} \left(   OPT_{\theta} - 2^{-I} \right) \geq \\
		& \geq \frac{OPT}{\rho} - 2^{-\frac{\rho}{2}+ 1 },
	\end{align*}
	where, in the second-last step, $OPT_{\theta}$ denotes the principal expected utility in an optimal contract for the non-Bayesian setting in which only type $\theta \in \Theta$ is present.
\end{proof}

The following theorem shows that the bounds provided in Theorems~\ref{thm:mult_add_no_bayes}~and~\ref{thm:mult_add_bayes} are tight.
%

\begin{theorem}\label{thm:mult_add_imp}
	No linear contract provides a $\big(  \rho ,2^{-O \left( \rho \right)} \big)$-approximation of an optimal contract, even in non-Bayesian principal-agent problems. Equivalently, for any $\rho \ge 1$, no linear contract provides a $\left(  \rho ,2^{- d \rho + e} \right)$-bi-approximation of an optimal one, for two constants $d \in \mathbb{N}$ and $e \in \mathbb{Z}$.
	%
\end{theorem}

\begin{proof}
	Given any $\rho \ge 1$, we show that there exists a non-Bayesian setting $(A,\Omega)$ in which, using a linear contract, it is impossible to obtain a $\left( \rho, 2^{-4 \rho -2} \right)$-approximation of the optimal expected utility for the principal.
	In these instances, an optimal (non-linear) contract provides the principal with an expected utility $OPT > 4\rho 2^{-\lfloor4 \rho\rfloor-2},$ while the best linear contract achieves at most $2^{-\lfloor 4  \rho\rfloor-1}$ utility.
	Since $\frac{ OPT}{\rho} -2^{-4 \rho -2}>2^{-\lfloor 4  \rho\rfloor-1}$, this concludes the proof.
	Formally, let us take $\Omega=\{\omega_1, \omega_2,\omega_3\}$ and $A=\{a_i\}_{i \in [\gamma]} \cup \{ \bar a \}$, where $\gamma \coloneqq \lfloor 4 \rho \rfloor$.
	Let $r_{\omega_1} = 1$ be the reward of outcome $\omega_1$, while the other outcomes provide zero reward, namely $r_{\omega_2} = r_{\omega_3} = 0$.
	The agent's actions are such that $F_{a_1,\omega_1}=\frac{1}{2}$ and $F_{a_1,\omega_2}=\frac{1}{2}$, while $F_{a_i,\omega_1}=2^{-i}$ and  $F_{a_i,\omega_3}=1-2^{-i}$ for every $i\in [\gamma]: i > 1$.
	Each action $a_i \in A$ has cost $c_{a_i}=2^{-i}-(\gamma-i+2)2^{-\gamma-2}$.
	Moreover, the last action $\bar a \in A$ is such that $F_{\bar a,\omega_3}=1$ and $c_{\bar a} = 0$ (ensuring IR).
	Simple arguments show that an optimal contract sets payments $p_{\omega_1} = p_{\omega_3} = 0$ and $p_{\omega_2}=1-(\gamma + 1) 2^{-\gamma-1}$, implementing action $a_1$ with a principal's expected utility $(\gamma+1) 2^{-\gamma-2}\ge 4\rho2^{-\lfloor 4  \rho \rfloor-2}$.
	Intuitively, the contract is such that $P_{a_1} = c_{a_1}$, which ensures that the agent plays action $a_1$ (it is IC and ties are broken in favor of the principal), while minimizing the payment.
	Next, we show that any linear contract provides the principal with an expected utility at most $2^{-\gamma-1}$.
	First, let us notice that the principal's expected utility when implementing action $a_\gamma$ is at most $R_{a_\gamma}-c_{a_\gamma}=2^{-\gamma-1}$.
	Instead, suppose that a linear contract implements an action $a_{i} \in A$ such that $i \in [\gamma] : 1 < i < \gamma$.
	Then, it must be the case that $a_{i}$ provides the agent with an expected utility greater than or equal to that obtained by $a_{i+1}$, \emph{i.e.}, it must be $2^{-i} p_{\omega_1}+ \left( 1-2^{-i} \right)	p_{\omega_3} -c_{a_i} \geq 2^{-i-1} p_{\omega_1} + \left( 1-2^{-i-1} \right)	p_{\omega_3} - c_{a_{i+1}}$. Thus: 
	\[
		2^{-i} p_{\omega_1} - \left( 2^{-i}-(\gamma-i+2) 2^{-\gamma-2} \right) \ge 2^{-i-1} p_{\omega_1} - \left( 2^{-i-1}-(\gamma-i+1) 2^{-\gamma-2} \right), 
	\]
	which implies that $p_{\omega_1}\ge 1-2^{i+1-\gamma-2}$.
	This prove that the expected utility for the principal is at most $\left( 1-p_{\omega_1} \right) 2^{-i}\le 2^{-\gamma-1}=2^{-\lfloor 4 \rho \rfloor-1} $, concluding the proof.
\end{proof}

\section{Linear versus Tractable Contracts}\label{sec:hardness}

In this section, we show that linear contracts have the same worst-case performance as efficiently-computable ones.
%
In Section~\ref{sec:hardness_mult}, we prove that there is no tractable contract providing a multiplicative loss sublinear in the number of agent's types.
This shows that, even if linear contracts provide a bad multiplicative loss in the number of agent's types (Theorem~\ref{thm:linear_lower_bound}), this is also true for tractable contracts.
%
%
Then, in Section~\ref{sec:hardness_bi_apx} we show that the $\big( \rho, 2^{-\Omega \left( \rho \right)} \big)$-bi-approximation result for linear contacts (Theorem~\ref{thm:mult_add_bayes}) is the best one can possibly achieve in polynomial time.
Technically, it is \textsf{NP}-hard to compute a contract providing a $\big(\rho, 2^{-\omega( \rho)}\big)$-bi-approximation of an optimal one.

\subsection{The Limits of Tractable Contracts}\label{sec:hardness_mult}

In the following Theorem~\ref{thm:hard_independent}, we prove that it is \textsf{NP}-hard to design a contract that approximates the overall principal's expected utility in an optimal contract up to within any multiplicative factor that is sublinear in the number of agent's types $\ell$.
The theorem is based on a reduction from GAP-INDEPENDENT-SET, which is the promise problem of deciding, in a given graph, whether there exists an independent set involving at least some (large) fraction of nodes or all the independent sets encompass at most some (small) fraction of nodes~\citep{Zuckerman2007linear}.

\begin{theorem}\label{thm:hard_independent}
	In Bayesian principal-agent problems, for any $\epsilon>0$ it is \textnormal{\textsf{NP}}-hard to design a contract providing a multiplicative loss $O \left( \ell^{1-\epsilon} \right)$ of an optimal one, where $\ell$ is the number of agent's types.
\end{theorem}

\begin{proof}
	We reduce from GAP-INDEPENDENT-SET, which is a promise problem that formally reads as follows: given $\epsilon > 0$ and a graph $G= (V,E)$, with set of nodes $V$ and set of edges $E$, determine whether $G$ admits an independent set of size at least $|V|^{1-\epsilon}$ or all the independent sets of $G$ have size smaller than $|V|^{\epsilon}$.
	The Bayesian principal-agent instances in the reduction have a number of agent's types $\ell = |V|$.
	The main idea of the proof is to show that, provided $\ell$ is large enough, if $G$ admits an independent set of size at least $\ell^{1-\epsilon}$, then in the corresponding principal-agent setting there exists a contract in which the overall principal's expected utility is at least $\frac{1}{2} \, \ell^{1-\epsilon} \, 2^{-\ell-1}$; otherwise, the utility is at most $2 \, \ell^{\epsilon} \, 2^{-\ell-1}$ for any contract.
	Since GAP-INDEPENDENT-SET is \textnormal{\textsf{NP}}-hard for every $\epsilon>0$~\citep{hastad1999clique,Zuckerman2007linear}, this is enough to prove the statement.
	
	\paragraph{Construction}
	Given a graph $G=(V,E)$, we build a Bayesian principal-agent setting $(\Theta, A, \Omega)$ as follows.
	For every node $v \in V$ in the graph $G$, there are two outcomes $\omega_v, \bar \omega_v \in \Omega$ such that $r_{\omega_v} = 1 $ and $r_{\bar \omega_v} = 0$.
	Moreover, there is an additional auxiliary outcome $\bar \omega \in \Omega$ with $r_{\bar \omega} = 0$.
	The agent type is uniformly selected from a set $\Theta = \{ \theta_v \}_{v \in V}$ of $\ell = |V|$ different types, each corresponding to a node in the graph.
	Thus, the distribution $\mu \in \Delta_{\Theta}$ is such that $\mu_{\theta_v} = \frac{1}{\ell}$ for every $\theta_v \in \Theta$.
	The agent has $m = \ell^2-\ell+1$ actions available.
	There is an action $\bar a \in A$ that induces a distribution over outcomes $F_{\theta_v, \bar a}$ with $F_{\theta_v, \bar a, \omega_v} =F_{\theta_v, \bar a, \bar \omega_v} = \frac{1}{2}$ and has cost $c_{\theta_v, \bar a} = \frac{1}{2} - \ell 2^{-\ell-1}$, no matter the agent's type $\theta_v \in \Theta$.
	Each of the remaining $\ell^2 - \ell$ actions, denoted as $a_{ui} \in A$, corresponds to a node $u \in V$ and an index $i \in [\ell-1]$.
	They are characterized by outcome distributions defined as follows.
	For every agent's type $\theta_v \in \Theta$ and action $a_{ui} \in A$ (with $v,u \in V$), the distribution $F_{\theta_v, a_{ui}}$ is such that:
	\begin{itemize}
		\item If $(v,u) \in E$, then the reachable outcomes are $\omega_v$, $\bar \omega_v$, $\bar \omega_{u}$, and $\bar \omega$, which are reached with probabilities, respectively, $F_{\theta_v,a_{ui}, \omega_v}=2^{-i-1}$, $F_{\theta_v,a_{ui},\bar \omega_v}=F_{\theta_v,a_{ui},\bar \omega_{u}}=\frac{2}{3}2^{-i-1}$, and $F_{\theta_v,a_{ui},\bar \omega}= 1-\frac{7}{3}2^{-i-1}$;
		\item If $(v,u) \notin E$, then the reachable outcomes are $\omega_v$ and $\bar \omega$, which are reached with probabilities $F_{\theta_v,a_{ui},\omega_v}=2^{-i-1}$ and $F_{\theta_v,a_{ui},\bar \omega}= 1-2^{-i-1}$, respectively.
	\end{itemize}
	Moreover, the cost of each action $a_{ui} \in A$ is $c_{\theta_v,a_{ui}}=2^{-i-1}-(\ell-i) 2^{-\ell-1}$, for any type $\theta_v \in \Theta$.

	\paragraph{Overview}
	In the instances of the reduction, the principal's expected utility contribution due to an agent's type playing an action $a_{u i} \in A$ is small.
	Thus, the principal's objective is to incentivize as many agent's types as possible to play $\bar a$, which is the action with the greatest difference between expected reward and cost.
	In order to incentivize an agent of type $\theta_v \in \Theta$ to play $\bar a$ rather than an action $a_{u 1} \in A$, while still achieving a satisfactory expected utility from that, the principal must set some (large) payment on outcome $\bar \omega_v$, some (small) payment on outcome $\omega_v$, and no payment on $\bar \omega$.
	Indeed, rewarding the last two outcomes prevents from reaching the desired principal's expected utility.
	Moreover, the principal must \emph{not} set payments on outcomes $\bar \omega_u$ such that vertex $u$ is adjacent to $v$, otherwise an agent of type $\theta_v \in \Theta$ would be incentivized to play action $a_{u 1}$ rather $\bar a$.
	This implies that the principal can extract a satisfactory utility only from agent's types whose corresponding vertices constitute an independent set of the graph $G$. 
	%

	\paragraph{Completeness}
	Suppose that graph $G$ admits an independent set of size at least $|V|^{1-\epsilon}$.
	Then, there exists a maximal independent set $V^\star \subseteq V$ of size $|V^\star| \geq |V|^{1-\epsilon}$ such that, for every node $v \notin V^\star$, there is a node $u \in V^\star$ with $(v,u) \in E$.
	Let us define a contract with $p_{\bar \omega_v} = 1-\ell \, 2^{-\ell-1}$ for all $v \in V^\star$ and $p_{ \omega_v}=\frac{1}{3} \left( 1-\ell 2^{-\ell-1} \right)+\ell 2^{-\ell-1}$ for all $v  \notin V^\star$, while all the other payments are set to $0$.
	First, we show that, given this contract, any agent of type $\theta_v \in \Theta$ with $v \in V^\star$ is incentivized to play action $\bar a$.
	The expected utility of the agent by playing $\bar a$ is:
	\[
		P_{\theta_v,\bar a} - c_{\theta_v, \bar a} = \frac{1}{2} \left( 1-\ell \, 2^{-\ell-1} \right)- \left( \frac{1}{2}-\ell \, 2^{-\ell-1} \right)=\frac{\ell}{2}\, 2^{-\ell-1}.
	\]
	As for the expected utility of playing an action $a_{u i} \in A$, two cases are possible.
	If $(v,u) \notin E$, then $P_{\theta_v, a_{ui}} = 0$ (since there is no payment associated to $\omega_v$, being $v \in V^\star$); thus, the resulting agent's expected utility is negative.
	Instead, if $(v,u) \in E$, then, by definition of independent set, it must be the case that $u \notin V^\star$, which implies that the only reachable outcome having non-zero payment is $\bar \omega_v$. Thus, in this case the agent's expected utility is:
	\[
		P_{\theta_v,a_{ui}} - c_{\theta_v,  a_{ui}} \hspace{-1mm}=\frac{2}{3} \ 2^{-i-1}  \left( 1-\ell 2^{-\ell-1} \right) -\left[ 2^{-i-1}\hspace{-1mm}-(\ell-i)2^{-\ell-1}) \right]\hspace{-1mm}= 2^{-\ell-1} \left( \ell-i- \frac{2}{3} \ \ell2^{-i-1}- \frac{1}{3} \ 2^{-i+\ell} \right).
	\]
	Then, for $i \geq \frac{\ell}{2}$, it immediately follows that:
	\[
		P_{\theta_v,a_{ui}} - c_{\theta_v,  a_{ui}}  = 2^{-\ell-1} \left( \ell-i-\frac{2}{3}\ \ell2^{-i-1}-\frac{1}{3} \ 2^{-i+\ell} \right) \le \frac{\ell}{2}2^{-\ell-1} = P_{\theta_v,\bar a} - c_{\theta_v, \bar a},
	\]
	while, for $i \leq \frac{\ell}{2}$, the same result follows from the fact that, provided $\ell$ is large enough, it holds:
	\[
		2^{-\ell-1}\left( \ell-i-\frac{2}{3}\ \ell2^{-i-1}-\frac{1}{3} \ 2^{-i+\ell} \right) \le    2^{-\ell-1} \left( \ell - \frac{1}{3} \ 2^{\frac{\ell}{2}} \right)\le \frac{\ell}{2}2^{-\ell-1} ,
	\]
	where the last inequality holds since $2^{\frac{\ell}{2}} \geq \frac{\ell}{2}$ for a sufficiently large $\ell$.
	This shows that any agent of type $\theta_v$ with $v \in V^\star$ plays action $\bar a$.
	Next, we prove that, whenever the agent's type $\theta_v \in \Theta$ is such that $v \notin V^\star$, then the agent plays an action $a_{u1}$ associated with a node $u \in V$ such that $(v,u) \in E$ and $u \in V^\star$.
	Notice that one such node always exists since $V^\star$ is a maximal independent set.
	Given that $p_{\bar \omega_v} = 0$, $p_{\bar \omega_{u}}= 1-\ell 2^{-\ell-1}$, and $p_{\omega_v}=\frac{1}{3} \left( 1-\ell 2^{-\ell-1} \right)+\ell 2^{-\ell-1}$, the expected utility of the agent by playing action $a_{u1}$ is:
	\begin{align*}
		P_{\theta_v, a_{u1}} - c_{\theta_v, a_{u1}} & = \frac{2}{3}2^{-2} \left( 1-\ell 2^{-\ell-1} \right) + 2^{-2} \left[ \frac{1}{3} \left( 1-\ell 2^{-\ell-1} \right) + \ell 2^{-\ell-1} \right]- \left[ 2^{-2}-(\ell-1)2^{-\ell-1} \right]=\\
		& = (\ell-1)2^{-\ell-1}.
	\end{align*}
	On the other hand, any action $a_{ui}$ provides the agent with an expected utility:
	\begin{align*}
		P_{\theta_v, a_{ui}} - c_{\theta_v, a_{ui}} &= \frac{2}{3}2^{-i-1} \left( 1-\ell 2^{-\ell-1} \right) 	+ \\
		& \textcolor{white}{=\qquad} + 2^{-i-1} \left[ \frac{1}{3} \left( 1-\ell 2^{-\ell-1} \right) + \ell 2^{-\ell-1} \right]- \left[ 2^{-i-1}-(\ell-i)2^{-\ell-1} \right] = \\
		& = (\ell-i)2^{-\ell-1}\le (\ell-1)2^{-\ell-1} = P_{\theta_v, a_{u1}} - c_{\theta_v, a_{u1}}.
	\end{align*}
	Moreover, it is easy to check that action $\bar a$ provides the agent with a negative utility, showing that any agent of type $\theta_v$ with $v \notin V^\star$ plays an action $a_{u1}$.
	Finally, we can conclude that the overall principal's expected utility is:
	\begin{align*}
		&\sum_{v \in V^\star } \mu_{\theta_v} \left( R_{\theta_v, \bar a} -P_{\theta_v, \bar a} \right) + \sum_{v \notin V^\star }  \mu_{\theta_v} \left( R_{\theta_v,  a_{u1}} -P_{\theta_v,  a_{u1}} \right)= \\
		&\qquad \,\, =\sum_{v \in V^\star } \frac{1}{\ell} \left[ \frac{1}{2}-\frac{1}{2} \left( 1-\ell 2^{-\ell-1} \right)\right] + \\
		&\qquad\,\, \textcolor{white}{=}\qquad + \sum_{v \notin V^\star } \frac{1}{\ell}  \left[  2^{-2} -\frac{2}{3}2^{-2} \left( 1-\ell 2^{-\ell-1} \right) - \frac{1}{3} 2^{-2}  \left( 1-\ell 2^{-\ell-1} \right) - 2^{-2} \ell 2^{-\ell-1}   \right] \geq \\
		&\qquad\,\, \geq \frac{1}{\ell} \ell^{1 -\epsilon}\left[ \frac{1}{2}-\frac{1}{2} \left( 1-\ell 2^{-\ell-1} \right)\right]  + \\ 
		&\qquad\,\, \textcolor{white}{=}\qquad +\frac{1}{\ell} \left( \ell - \ell^{1-\epsilon} \right) \left[  2^{-2} -\frac{2}{3}2^{-2} \left( 1-\ell 2^{-\ell-1} \right) - \frac{1}{3} 2^{-2}  \left( 1-\ell 2^{-\ell-1} \right) - 2^{-2} \ell 2^{-\ell-1}   \right]  = \\
		&\qquad\,\, = \frac{1}{\ell} \ell^{1 -\epsilon} \frac{\ell}{2}2^{-\ell-1} = \frac{1}{2}\ell^{1-\epsilon} 2^{-\ell-1},
	\end{align*}
	where we used the fact that $|V^\star| \geq |V|^{1-\epsilon} = \ell^{1-\epsilon}$ by assumption.

	\paragraph{Soundness}
	We start showing that, if the principal deploys a contract that implements an action $a_{ui} \in A$, then she achieves an expected utility $R_{\theta_v, a_{ui}} - P_{\theta_v, a_{ui}} \leq 2^{-\ell}$, no matter the agent's type $\theta_v \in \Theta$.
	First, notice that, by implementing actions $a_{u \, \ell-1} \in A$, the principal can obtain an expected utility at most of $R_{\theta_v, a_{u \, \ell-1}} - c_{\theta_v, a_{u \, \ell-1}} = 2^{-\ell}-2^{-\ell-1}=2^{-\ell-1}$ (due to IR constraints).
	Next, suppose that the contract implements an action $a_{u i} \in A$ with $i \in [\ell-2]$ for an agent of type $\theta_v \in \Theta$.
	Then, by IC constraints, action $a_{u i}$ must provide the agent with an expected utility greater than or equal to that achieved by playing $a_{u \, i+1}$, \emph{i.e.}, it must be $P_{\theta_v, a_{ui}} - c_{\theta_v, a_{ui}} \geq P_{\theta_v, a_{u \, i+1}} - c_{\theta_v, a_{u \, i+1}}$.
	Two cases are possible.
	In the first one, it holds $(v,u) \in E$, which implies that:
	\begin{align*}
			&p_{\omega_v} 2^{-i-1} + p_{\bar \omega_v} \frac{2}{3} 2^{-i-1} + p_{\bar \omega_u} \frac{2}{3} 2^{-i-1}	+ p_{\bar \omega} \left(  1 - \frac{7}{3} 2^{-i-1} \right) - \left[  2^{-i-1} - \left( \ell-i \right) 2^{-\ell-1} \right] \geq \\
			& \,\, \geq p_{\omega_v} 2^{-i-2} + p_{\bar \omega_v} \frac{2}{3} 2^{-i-2} + p_{\bar \omega_u} \frac{2}{3} 2^{-i-2}	+ p_{\bar \omega} \left(  1 - \frac{7}{3} 2^{-i-2} \right) - \left[  2^{-i-2} - \left( \ell-i-1 \right) 2^{-\ell-1} \right].
	\end{align*}
	Thus,
	\[
		2^{-i-1} \left(   p_{\omega_v} + \frac{2}{3} p_{\bar \omega_v} + \frac{2}{3} p_{\bar \omega_{u}} - \frac{7}{3} p_{\bar \omega} \right) \geq 2^{-i-2} - 2^{-\ell-1}.
	\]
	As a result, the expected utility of the principal when an agent of type $\theta_v$ plays an action $a_{u i}$ is:
	\begin{align*}
		R_{\theta_v, a_{ui}} - P_{\theta_v, a_{ui}} & = 2^{-i-1} - \left[ p_{\omega_v} 2^{-i-1} + p_{\bar \omega_v} \frac{2}{3} 2^{-i-1} + p_{\bar \omega_u} \frac{2}{3} 2^{-i-1}	+ p_{\bar \omega} \left(  1 - \frac{7}{3} 2^{-i-1} \right)    \right] \leq \\
		& \leq 2^{-i-1} - \left( p_{\omega_v} 2^{-i-1} + p_{\bar \omega_v} \frac{2}{3} 2^{-i-1} + p_{\bar \omega_u} \frac{2}{3} 2^{-i-1}	-p_{\bar \omega} \frac{7}{3} 2^{-i-1}     \right) \leq \\
		& \leq 2^{-i-1} - 2 \left(  2^{-i-2} - 2^{-\ell-1} \right) = 2^{-\ell}.
	\end{align*}
	A similar argument holds for the case in which $(v,u) \notin E$.
	Now, given a contract, let $\Theta^\star \subseteq \Theta$ be the set of agent's types $\theta_v$ such that: \emph{(i)} the contract implements action $\bar a$ for an agent of type $\theta_v$; and \emph{(ii)} the principal achieves an expected utility strictly larger than $2^{-\ell}$ when an agent of type $\theta_v$ plays $\bar a$.
	We prove that, for any contract, any pair of types $\theta_v, \theta_u \in \Theta^\star$ is such that $(v,u) \notin E$.
	By contradiction, suppose that there exist $\theta_v, \theta_u \in \Theta^\star$ such that $(u,v) \in E$.
	We distinguish two cases.
	The first one is when $p_{\bar \omega_v} \le p_{\bar \omega_{u}}$.
	Since action $\bar a $ must provide an agent of type $\theta_v$ with an expected utility greater than or equal to that obtained for action $a_{u1}$ (by IC constraints), we have that:
	\[
		\frac{1}{2} p_{\omega_v} + \frac{1}{2} p_{\bar \omega_v} - \left( \frac{1}{2} - \ell 2^{-\ell-1}  \right)\geq 2^{-2} p_{\omega_v} + \frac{2}{3} 2^{-2}  \left(  p_{\bar \omega_v} + p_{\bar \omega_{u}} \right) + \left( 1 - \frac{7}{3} 2^{-2} \right) p_{\bar \omega} - \left[  2^{-2} - \left( \ell-1 \right) 2^{-\ell-1}\right].
	\]
	By using the fact that $p_{\bar \omega_u} \geq p_{\bar \omega_{v}} \geq 0$ and $p_{\bar \omega} \geq 0$, and re-arranging the terms, we obtain that $2^{-2} \left(  p_{\omega_{v}} + p_{\bar \omega_{v}} \right) \geq 2^{-2} - 2^{-\ell-1}$, which implies that $P_{\theta_v, a_{\bar v}} = \frac{1}{2} p_{\omega_{v}} + \frac{1}{2} p_{\bar \omega_{v}} \geq \frac{1}{2} - 2^{-\ell}$.
	Since $R_{\theta_v, \bar a} = \frac{1}{2}$, this results in a principal's expected utility at most of $2^{-\ell}$, which is a contradiction.
	In the second case in which $p_{\bar \omega_v} \geq p_{\bar \omega_{u}}$, we reach a contradiction using an analogous argument for an agent of type $\theta_u$ (rather than $\theta_v$).
	Thus, we can conclude that, for any contract, the set of nodes $v \in V$ such that $\theta_v \in \Theta^\star$ constitutes an independent set of the graph $G$.
	Moreover, notice that the maximum expected utility that the principal can obtain when an agent of type $\theta_v \in \Theta$ plays action $\bar a$ is $R_{\theta_v, \bar a} - c_{\theta_v, \bar a} = \ell 2^{-\ell-1}$.
	Since, by assumption, the largest independent set of $G$ has size at most $|V|^\epsilon = \ell^\epsilon$, we can conclude that in any contract the overall expected utility of the principal is:
	\begin{align*}
		\sum_{\theta_v \in \Theta^\star} \mu_{\theta_v} \left( R_{\theta_v, \bar a} - P_{\theta_v, \bar a}  \right) + \sum_{\theta_v \in \Theta \setminus \Theta^\star} \mu_{\theta_v} \left( R_{\theta_v,  a^*(\theta_v)} - P_{\theta_v,  a^*(\theta_v)}  \right) & \leq  \frac{1}{\ell} \ell^\epsilon \ell 2^{-\ell-1} + \frac{1}{\ell} \left( \ell -\ell^\epsilon  \right) 2^{-\ell} \leq \\
		& \leq 2\,  \frac{1}{\ell} \, \ell^{1+\epsilon}  \, 2^{-\ell-1} = \\
		& = 2 \, \ell^\epsilon \, 2^{-\ell-1},
	\end{align*}
	where the last inequality holds provided that $\ell$ is sufficiently large.
\end{proof}

\subsection{The Limits of Bi-Approximations}\label{sec:hardness_bi_apx}

We show that, for any $\rho\ge 1$, it is \textnormal{\textsf{NP}}-hard to design a contract providing a $\big(\rho,2^{-\omega(\rho)} \big)$-bi-approximation of an optimal one.
To this end, we employ a reduction from a {promise problem} associated with \textsf{LABEL-COVER} instances, whose definition follows.

\begin{definition}[\textnormal{\textsf{LABEL-COVER}} instance]
	An instance of \textnormal{\textsf{LABEL-COVER}} is a tuple $( G, \Sigma, \Pi)$:
	\begin{itemize}
		\item $G \coloneqq (U,V,E)$ is a \emph{bipartite graph} defined by two disjoint sets of nodes $U$ and $V$, connected by the edges in $E \subseteq U \times V$, which are such that all the nodes in $U$ have the same degree;
		\item $\Sigma$ is a finite set of \emph{labels}; and
		\item $\Pi \coloneqq \left\{ \Pi_e : \Sigma \to \Sigma \mid e \in E \right\}$ is a finite set of \emph{edge constraints}.
	\end{itemize}
	Moreover, a \emph{labeling} of the graph $G$ is a mapping $\pi: U \cup V \to \Sigma$ that assigns a label to each vertex of $G$ such that all the edge constraints are satisfied.
	Formally, a labeling $\pi$ satisfies the constraint for an edge $e =(u,v) \in E$ if it holds that $ \pi(v) = \Pi_e(\pi(u))$.
\end{definition}

The classical \textsf{LABEL-COVER} problem is the search problem of finding a valid labeling for a \textsf{LABEL-COVER} instance given as input.
In the following, we consider a different version of the problem, which is the {promise problem} associated with \textsf{LABEL-COVER} instances.

\begin{definition}[\textnormal{\textsf{GAP-LABEL-COVER}$_{c,s}$}]
	For any pair of numbers $0 < s < c < 1$, we define \textnormal{\textsf{GAP-LABEL-COVER}}$_{c,s}$ as the following promise problem.
	\begin{itemize}
		\item \textnormal{\texttt{Input:}} An instance $(G,\Sigma, \Pi)$ of \textnormal{\textsf{LABEL-COVER}} such that either one of the following is true:
		\begin{itemize}
			\item there exists a labeling $\pi$ that satisfies at least a fraction $c$ of the edge constraints in $\Pi$;
			\item any labeling $\pi$ satisfies less than a fraction $s$ of the edge constraints in $\Pi$.
		\end{itemize}
		\item \textnormal{\texttt{Output:}} Determine which of the above two cases hold.
	\end{itemize}
\end{definition}

In order to prove Theorem~\ref{thm:hard_label}, we use the following result due to~\citet{raz1998parallel}~and~\citet{arora1998proof}.

\begin{theorem}[\citet{raz1998parallel,arora1998proof}]\label{thm:hard_gap_label}
	For any $\epsilon > 0$, there exists a constant $k_\epsilon \in \mathbb{N}$ that depends on $\epsilon$ such that the promise problem \textnormal{\textsf{GAP-LABEL-COVER}}$_{1,\epsilon}$ restricted to inputs $(G, \Sigma, \Pi)$ with $|\Sigma| = k_\epsilon$ is \textnormal{\textsf{NP}}-hard.
\end{theorem}

Next, we show our main result.~\footnote{
In order to prove Theorem~\ref{thm:hard_label}, we need that the difference between the overall principal's expected utility in the completeness part and that in the soundness part is at least $2^{-O(\rho)}$, otherwise a contract providing a $\big(\rho, 2^{-O(\rho)}\big)$-bi-approximation cannot distinguish between the two cases.
Thus, we cannot use the construction in Theorem~\ref{thm:hard_independent}, since it does \emph{not} enjoy this property.
Indeed, we would like that the principal's expected utility in the soundness case decreases at a rate of $2^{-O(\rho)}$ as $\rho$ increases, while in Theorem~\ref{thm:hard_independent} the principal's expected utility decreases with the number of agent's types, \emph{i.e.}, its maximum value is $2^{-\ell}$.
Moreover, in Theorem \ref{thm:hard_independent} we reduce from \textsf{GAP-INDEPENDENT-SET}, which has \emph{not} perfect completeness.
Thus, the principal can extract a satisfactory utility from at most a fraction $\ell^{-\epsilon}$ of the agent's types, which implies that the expected utility decreases with the number of agent's types.
In order to deal with these problems, we base our reduction on $\textsf{GAP-LABEL-COVER}_{c,s}$.
Using this problem, we have perfect completeness, though at the expense of the \textsf{NP}-hardness of approximating only to within any multiplicative constant factor.
This is sufficient for proving Theorem~\ref{thm:hard_label}, since it requires the \textsf{NP}-hardness of approximating up to within a multiplicative factor that is of the order of $\Theta(\rho)$.}

\begin{theorem} \label{thm:hard_label}
	Given a Bayesian principal-agent setting, it is \textnormal{\textsf{NP}}-hard to design a contract providing a $\big(\rho,2^{-\omega(\rho)} \big)$-bi-approximation of an optimal one. Equivalently, for any $\rho \ge 1$, it is \textnormal{\textsf{NP}}-hard to design a contract providing a $\big(  \rho ,2^{- d \rho + e} \big)$-bi-approximation for two constants $d \in \mathbb{R}^+$, $e \in \mathbb{R}$.
\end{theorem}

\begin{proof}
	Letting $\gamma \coloneqq \lceil {10\rho} \rceil$, we prove the result by means of a reduction from \textsf{GAP-LABEL-COVER}$_{1,\frac{1}{2\gamma}}$.
	In particular, our construction is such that, if the \textsf{LABEL-COVER} instance admits a labeling that satisfies all the edge constraints (recall that $c=1$), then the corresponding Bayesian principal-agent setting admits a contract providing the principal with an overall expected utility at least of $(\gamma+2)2^{\gamma-4}$.
	Instead, if at most a $\frac{1}{2\gamma}$ fraction of the edge constraints are satisfied by any labeling, then the principal expected utility is at most $2^{-\gamma-1}$ in any contract.
	By Theorem~\ref{thm:hard_gap_label}, this implies that designing a contract giving a $\big( \rho,2^{-8 \rho-4} \big)$-bi-approximation (for any $\rho \ge 1$) is \textsf{NP}-hard.
	%
	Indeed, the following relation shows that a $\big( \rho,2^{-8 \rho-4} \big)$-bi-approximation algorithm can determine whether the \textsf{LABEL-COVER} instance admits a labeling that satisfies all the edge constraints or at most a $\frac{1}{2\gamma}$ fraction of the edge constraints are satisfied by any labeling:
	\[
		\frac{1}{\rho} \left( \gamma+2 \right) 2^{\gamma-4}- 2^{-8\rho-4} \ge \frac{1}{\rho} \left( \gamma+2 \right) 2^{\gamma-4} -2^{-\gamma-3}> 10 \cdot \ 2^{\gamma-4}- 2^{-\gamma-3}\ge 2^{-\gamma-1}.
	\]
	Next, we provide the formal definition of our reduction and prove its crucial properties.
	
	\paragraph{Construction}
	Given an instance of \textsf{LABEL-COVER} $(G,\Sigma,\Pi)$ with a bipartite graph $G = (U,V,E)$, we build a Bayesian principal-agent setting $(\Theta,A,\Omega)$ as follows.
	For every node $v \in U \cup V$ of $G$ and label $\sigma \in \Sigma$, there is an outcome $\omega_{v  \sigma} \in \Omega$ with reward $r_{\omega_{v  \sigma} } = 0$ to the principal.
	Moreover, there are two additional outcomes $\omega_0, \omega_1 \in \Omega$ such that $r_{\omega_0} = 0$ and $r_{\omega_1} = 1$.
	The agent can be of $\ell = |E|$ different types, each associated with an edge of $G$; formally, $\Theta = \{ \theta_e \}_{e \in E}$.
	All the types have the same probability of occurring, being $\mu \in \Delta_\Theta$ such that $\mu_{\theta_e} = \frac{1}{\ell}$ for $e \in E$.
	For the ease of presentation and w.l.o.g., we let each agent's type $\theta_e \in \Theta$ having a different action set $A_{\theta_e}$, so that, with an abuse of notation, $A = \{ A_{\theta_e} \}_{\theta_e \in \Theta}$.
	Notice that, in order to recover a principal-agent setting in which each agent's type has the same set of actions, it is sufficient to add some dummy actions having zero cost for the agent and deterministically leading to outcome $\omega_0$ (with zero reward).
	Each agent's type $\theta_e \in \Theta$ with $e = (u,v)$ has an action $a_{\sigma  \sigma'} \in A_{\theta_e}$ for every pair of labels such that $\sigma \in \Sigma$ and $\sigma' = \Pi_e(\sigma)$.
	The action induces a probability distribution over outcomes $F_{\theta_e, a_{\sigma  \sigma'} }$ such that:
	\begin{itemize}
		\item Outcome $\omega_1$ is reached half of the times, being $F_{\theta_e, a_{\sigma  \sigma'}, \omega_1 } = \frac{1}{2} $;
		\item In the other half of the cases, outcomes $\omega_{u  \sigma}$ and $\omega_{v  \sigma'}$ are reached with equal probability, being $F_{\theta_e, a_{\sigma  \sigma'}, \omega_{u  \sigma} } = F_{\theta_e, a_{\sigma  \sigma'}, \omega_{v  \sigma'} } = \frac{1}{4} $.
	\end{itemize}
	The cost of the action is $c_{\theta_e, a_{\sigma \sigma'}} = \frac{1}{2}-(\gamma+2) 2^{-\gamma-3}$, no matter the agent's type $\theta_e \in \Theta$.
	Moreover, each agent's type $\theta_e \in \Theta$ with $e = (u,v)$ has an action $a_{i \sigma  \sigma'} \in A_{\theta_e}$ for every index $i \in [\gamma]$ and pair of labels $\sigma, \sigma' \in \Sigma$ such that $ \sigma' \neq  \Pi_e(\sigma) $.
	The action probability distribution $F_{\theta_e, a_{i \sigma  \sigma'} }$ is such that:
	\begin{itemize}
		\item Outcomes  $\omega_{u  \sigma}$ and $\omega_{v  \sigma'} $ are reached with the same (small) probability decreasing exponentially in the value of $i$, being $F_{\theta_e, a_{i \sigma  \sigma'}, \omega_{u  \sigma} } = F_{\theta_e, a_{i \sigma  \sigma'}, \omega_{v  \sigma'} } = 2^{-i-2}$;
		\item Outcome $\omega_1$ is reached with a probability twice as large as that of the previous ones, as $F_{\theta_e, a_{i \sigma \sigma'} , \omega_1 } = 2^{-i-1}$; 
		\item In all the other cases outcome $\omega_0$ is reached, since $F_{\theta_e, a_{i \sigma \sigma'} , \omega_0 } = 1 - 2^{-i}$.
	\end{itemize}
	Finally, the cost of the action is $c_{\theta_e, a_{i \sigma \sigma'} } = 2^{-i-1}-(\gamma-i+2) 2^{-\gamma-3}$
	
	\paragraph{Overview}
	The Bayesian principal-agent instances of the reduction have a structure similar to those in the proof of Theorem~\ref{thm:hard_independent}.
	Here, the contribution to the overall principal's expected utility due to an agent's type playing an action $a_{i \sigma \sigma'} \in A$ is small.
	Thus, the principal's objective is to incentivize as many agent's types as possible to play an action $a_{\sigma\sigma'}$.
	We recall that, for each agent's type $\theta_{e} \in \Theta$ with $e = (u,v)$, there exists an action $a_{\sigma \sigma'}$ only if the labels $\sigma$ and $\sigma'$ satisfy the constraint for edge $e$, namely $\sigma' = \Pi_e(\sigma)$.
	Moreover, in order for the principal to incentivize an agent's type to play $a_{\sigma \sigma'}$ and extract a satisfactory utility from that, the principal must commit to a contract that sets some payments on outcomes $\omega_{u \sigma}$ and $\omega_{v \sigma'}$.
	More precisely, an agent of type $\theta_{e} \in \Theta$ with $e = (u,v)$ is incentivized to play $a_{\sigma \sigma'}$ if the payments on outcomes $\omega_{u \sigma}$ and $\omega_{v \sigma'}$ are equal and sufficiently large.
	At the same time, there must \emph{not} be two labels $\sigma_u, \sigma_v \in \Sigma$ with $\sigma_u \neq \sigma$ and $\sigma_v \neq \sigma'$ such that a large payment is assigned to either $\omega_{u \sigma''}$ or $\omega_{v \sigma''}$, otherwise an agent of type $\theta_e$ would be incentivized to play action $a_{1 \sigma_u \sigma_v}$ rather than $a_{\sigma \sigma'}$.
	Then, for every vertex $v \in U \cup V$ of the graph $G$, there exists a single label $\sigma \in \Sigma$ such that there is some payment on $\omega_{v \sigma}$ and these labels define a labeling that satisfies all the constraints of edges corresponding to agent's types that play action $a_{\sigma \sigma'}$ while resulting in a satisfactory principal's expected utility.
	%

	\paragraph{Completeness}
	Suppose the instance of \textsf{LABEL-COVER} $(G,\Sigma,\Pi)$ admits a labeling $\pi : U \cup V \to \Sigma$ that satisfies all the edge constraints in $\Pi$.
	Let us define a contract such that $p_{\omega_{v \pi(v)}} = 1-(\gamma+2) 2^{-\gamma-3}$ for every node $v \in U \cup V$, while all the other payments are set to zero.
	First, we show that, given this contract, an agent of type $\theta_{e} \in \Theta$ with $e= (u,v)$ is incentivized to play action $a_{\pi(u) \pi(v)}$.
	Recall that, in our construction, an agent of type $\theta_{e}$ has action $a_{\pi(u) \pi(v)}$ available if and only if $\pi(v) = \Pi_e(\pi(u))$, which is always true since the labeling $\pi$ satisfies all the edge constraints by assumption.
	Given the definition of the contract, it holds that $p_{\omega_{u \pi(v)}} = p_{\omega_{v \pi(v)}} = 1-(\gamma+2) 2^{-\gamma-3}$, while $p_{\omega_{u \sigma}} = 0$ for every $\sigma \in \Sigma \setminus \{ \pi(u) \}$ and $p_{\omega_{v \sigma}} = 0$ for every $\sigma \in \Sigma \setminus \{ \pi(v) \}$.
	This implies that the expected utility of an agent of type $\theta_{e}$ by playing action $a_{\pi(u) \pi(v)}$ is:
	\begin{align*}
		P_{\theta_{e}, a_{\pi(u) \pi(v)}} - c_{\theta_e, a_{\pi(u) \pi(v)}} & =\frac{1}{4} \left( p_{\omega_{u\sigma}} +p_{\omega_{v\sigma'}} \right) - \left[ \frac{1}{2}- (\gamma+2) 2^{-\gamma-3} \right]= \\
		& =-(\gamma+2) 2^{-\gamma-4} +(\gamma+2)2^{-\gamma-3}= \\
		& =(\gamma+2) 2^{-\gamma-4}.
	\end{align*}
	Moreover, for any pair of labels $\sigma, \sigma' \in \Sigma$ such that $ \sigma' \neq  \Pi_e(\sigma) $, each action $a_{i \sigma \sigma'}$ for $i \in [\gamma]$ provides an expected utility of:
	\begin{align*}
		P_{\theta_{e}, a_{i \sigma \sigma'}} - c_{\theta_e, a_{i \sigma\sigma'}} & =2^{-i-2} \left[ 1-(\gamma+2) 2^{-\gamma-3} \right] - \left[ 2^{-i-1}-(\gamma-i+2)2^{-\gamma-3} \right]= \\
		& = 2^{-\gamma-3} \left[ -2^{-i-2+\gamma+3}-(\gamma+2) 2^{-i-2}+ \gamma-i+2 \right] ,
	\end{align*}
	which holds since it cannot be the case that both $p_{\omega_{u \sigma}}$ and $p_{\omega_{v \sigma'}}$ are different from zero, otherwise it would be $\pi(v) \in \Sigma \setminus \{ \pi(u) \}$, contradicting the fact that the labeling $\pi$ satisfies all the edge constraints.
	We distinguish two cases.
	In the first one, it holds $i\ge \frac{\gamma}{2}+1$.
	Then,
	\[
		P_{\theta_{e}, a_{i \sigma \sigma'}} - c_{\theta_e, a_{i \sigma\sigma'}}  \le 2^{-\gamma-3}  (\gamma-i+2) \le (\gamma+2) 2^{-\gamma-4} = P_{\theta_{e}, a_{\pi(u) \pi(v)}} - c_{\theta_e, a_{\pi(u) \pi(v)}}.
	\]
	In the second case, it holds $i\le \frac{\gamma}{2}+1$, which implies that:
	\[
		P_{\theta_{e}, a_{i \sigma \sigma'}} - c_{\theta_e, a_{i \sigma\sigma'}}  \leq 2^{-\gamma-3} (\gamma+2-2^{\frac{\gamma}{2}})\le \frac{\gamma}{2}2^{-\gamma-3} \leq P_{\theta_{e}, a_{\pi(u) \pi(v)}} - c_{\theta_e, a_{\pi(u) \pi(v)}},
	\]
	where the second-last inequality holds since $\frac{\gamma}{2} + 2 \leq 2^{\frac{\gamma}{2}}$ for $\gamma \geq 4$.
	Finally, it is easy to see that all the actions $a_{\sigma \sigma'}$ that are different from $a_{\pi(u) \pi(v)}$ provide an agent of type $\theta_{e}$ with an expected utility smaller than that achieved by playing $a_{\pi(u) \pi(v)}$.
	This shows that the contract incentivizes each agent's type $\theta_{e} \in \Theta$ with $e = (u,v)$ to play action $a_{\pi(u) \pi(v)}$.
	In conclusion, the overall expected utility of the principal is:
	\begin{align*}
		\sum_{\theta_{e} \in \Theta} \mu_{\theta_{e}} \left( R_{\theta_{e}, a^*(\theta_{e})} - P_{\theta_{e}, a^*(\theta_{e})} \right) &= \frac{1}{\ell} \sum_{e = (u,v) \in E} R_{\theta_{e}, a_{\pi(u) \pi(v)} } - P_{\theta_{e}, a_{\pi(u) \pi(v)} } = \\
		& = \frac{1}{2} - \frac{1}{2} \left[  1-(\gamma+2) 2^{-\gamma-3} \right] =\\
		& = (\gamma+2) 2^{-\gamma-4}.
	\end{align*}
	
	\paragraph{Soundness}
	We show that, if the \textsf{LABEL-COVER} instance is such that every labeling $\pi : U \cup V \to \Sigma$ satisfies at most a fraction $\frac{1}{2\gamma}$ of the edge constraints in $\Pi$, then, in the corresponding principal-agent setting, any contract provides the principal with an expected utility at most of $2^{-\gamma-1}$.
	As a first step, we show that all the actions $a_{i \sigma \sigma'}$ provide the principal with an expected utility at most of $2^{-\gamma-2}$.
	Assume that the agent has type $\theta_{e} \in \Theta$ with $e = (u,v)$ and that the contract deployed by the principal implements an action $a_{\gamma \sigma \sigma'}$ for an agent of type $\theta_{e}$, for some $\sigma, \sigma' \in \Sigma$ such that $\sigma' \neq \Pi_e(\sigma)$.
	Then, the principal's expected reward is $R_{\theta_{e}, a_{\gamma \sigma \sigma'}} = 2^{-\gamma-1}$, while the agent's cost is $c_{\theta_{e}, a_{\gamma \sigma \sigma'}} = 2^{-\gamma-1}-2^{-\gamma-2}$, implying that the principal's expected utility is at most $2^{-\gamma-2}$.
	Now, assume that the contract implements an action $a_{i \sigma \sigma'}$ with $i \in [\gamma] : i < \gamma$ for an agent of type $\theta_{e}$, for some $\sigma, \sigma' \in \Sigma$ such that $\sigma' \neq \Pi_e(\sigma)$.
	Then, since the action $a_{i \sigma \sigma'}$ must be IC, it must provide the agent with an expected utility greater than or equal to that provided by action $a_{i+1 \, \sigma \sigma'}$.
	Thus, it must be the case that $P_{\theta_e, a_{i \sigma \sigma'}} - c_{\theta_e, a_{i \sigma \sigma'}} \geq P_{\theta_e, a_{i+1 \, \sigma \sigma'}} - c_{\theta_e, a_{i+1 \,  \sigma \sigma'}} $, which implies that:
	\begin{align*}
		&2^{-i-1} p_{\omega_1} + 2^{-i-2} p_{\omega_{u\sigma}} +2^{-i-2} p_{\omega_{v \sigma'}} + \left( 1-2^{-i} \right) p_{ \omega_0}    - \left[ 2^{-i-1}-(\gamma-i+2) 2^{-\gamma-3} \right]\ge\\
		& \geq 2^{-i-2} p_{\omega_1}+ 2^{-i-3} p_{\omega_{u \sigma}} +2^{-i-3} p_{\omega_{v \sigma'}} +  \left( 1-2^{-i-1} \right) p_{\omega_0}  - \left[ 2^{-i-2}-(\gamma-i+1) 2^{-\gamma-3} \right].
	\end{align*}
	Thus, by re-arranging the terms and using the fact that $p_{\omega_0} \geq 0$, we get $2 p_{\omega_1}+ p_{\omega_{u \sigma}} +p_{\omega_{v \sigma'}} \ge 2- 2^{-\gamma+i}$, which implies that the principal's expected utility is:
	\[
		R_{\theta_{e}, a_{i \sigma \sigma'} } - P_{\theta_e, a_{i \sigma \sigma'} } = 2^{-i-1} - 2^{-i-1} p_{\omega_1} - 2^{-i-2} p_{\omega_{u\sigma}} -2^{-i-2} p_{\omega_{v \sigma'}}- \left( 1-2^{-i} \right) p_{ \omega_0}    \leq 2^{-\gamma-2}.
	\]
	This proves that any agent's action $a_{i \sigma \sigma'}$ provides the principal with an expected utility at most of $2^{-\gamma-2}$.
	Next, we switch the attention to actions $a_{\sigma \sigma'}$.
	Given a contract, let $\pi : U \cup V \to \Sigma$ be a labeling for the \textsf{LABEL-COVER} instance such that $\pi(v) \in \argmax_{\sigma \in \Sigma} p_{\omega_{v \sigma}}$ for every $v \in U \cup V$ (with ties broken arbitrarily).
	We show that, for an agent of type $\theta_e \in \Theta$, the contract implements an action providing the principal with an expected utility greater than $2^{-\gamma-2}$ only if the labeling $\pi$ satisfies the constraint $\Pi_e$ associated with edge $e$.
	By contradiction, suppose that $e = (u,v)$ and the constraint $\Pi_e$ is \emph{not} satisfied by $\pi$ since $\pi(v) \neq \Pi_e(\pi(u))$.
	Then, there is an agent's action $a_{1 \sigma_u \sigma_v} \in A_{\theta_e}$ with $\sigma_u \in \argmax_{\sigma \in \Sigma} p_{\omega_{\omega_{u \sigma} }  }$ and $\sigma_v \in \argmax_{\sigma \in \Sigma} p_{\omega_{\omega_{v \sigma} }  }$ such that the agent's expected utility is:
	\[
		P_{\theta_e, a_{1 \sigma_u \sigma_v}} - c_{\theta_e, a_{1 \sigma_u \sigma_v}} = \frac{1}{4} p_{\omega_1} + \frac{1}{8} \left(  p_{\omega_{u \sigma_u} }  +  p_{\omega_{v \sigma_v} }  \right) + \frac{1}{2} p_{\omega_0} - \left[ \frac{1}{4} - (\gamma+1) 2^{-\gamma - 3} \right].
	\]
	Moreover, all the actions $a_{\sigma \sigma'} \in A_{\theta_e}$ for $\sigma \in \Sigma$ and $\sigma' = \Pi_e(\sigma)$ provide the agent with a utility:
	\[
		P_{\theta_e, a_{ \sigma \sigma'}} - c_{\theta_e, a_{ \sigma \sigma'}} =\frac{1}{2}p_{\omega_1}+\frac{1}{4} \left( p_{\omega_{u \sigma}}+p_{\omega_{v \sigma'}} \right)- \left[ \frac{1}{2}-(\gamma+2) 2^{-\gamma-3}\right],
	\]
	where, by definition, it holds $p_{\omega_{u \sigma}}\le p_{\omega_{u\sigma_u}}$ and $p_{\omega_{v \sigma'}}\le p_{\omega_{v \sigma_v}}$.
	Thus, since an action $a_{\sigma \sigma'}$ is IC only if it holds that $P_{\theta_e, a_{ \sigma \sigma'}} - c_{\theta_e, a_{ \sigma \sigma'}} \geq P_{\theta_e, a_{1 \sigma_u \sigma_v}} - c_{\theta_e, a_{1 \sigma_u \sigma_v}} $, we can conclude that $2p_{\omega_1} +p_{\omega_{u \sigma}}+p_{\omega_{v \sigma'}} \ge 2- 2^{-\gamma}$.
	As a result, the expected utility of the principal is:
	\[
		R_{\theta_e, a_{ \sigma \sigma'}  } - P_{\theta_e , a_{ \sigma \sigma'}  } = \frac{1}{2} - \frac{1}{2}p_{\omega_1} - \frac{1}{4} \left( p_{\omega_{u \sigma}}+p_{\omega_{v \sigma'}} \right) \leq 2^{-\gamma-2},
	\]
	which is a contradiction.
	Finally, the maximum expected utility the principal can achieve for an agent of any type $\theta_e \in \Theta$ is $\max_{a \in A_{\theta_e}} \left\{ R_{\theta_e, a} - c_{\theta_e, a} \right\} = (\gamma+2)2^{-\gamma-3}$.
	By assumption, any labeling satisfies at most a fraction $\frac{1}{2\gamma} |E|$ of the edge constraints, thus, given any contract, at most a fraction $\frac{1}{2\gamma} \ell$ of agent's types play an action providing the principal with an expected utility greater than $2^{-\gamma-2}$ (and at most $(\gamma+2)2^{-\gamma-3}$).
	Then, the overall principal's expected utility in any contract is:
	\[
		\sum_{\theta_{e} \in \Theta} \mu_{\theta_{e}} \left( R_{\theta_{e}, a^*(\theta_{e})} - P_{\theta_{e}, a^*(\theta_{e})} \right) < \frac{1}{\ell} \cdot \frac{\ell}{2 \gamma} (\gamma+2)2^{-\gamma-3}+ \frac{1}{\ell} \left( \ell - \frac{\ell}{2 \gamma} \right)2^{-\gamma-2} \leq 2^{-\gamma+1},
	\]
	which concludes the proof.
\end{proof}

\section{Tractable Cases}\label{sec:tractable}

In this section, we investigate under which circumstances the problem of finding an optimal contract in Bayesian principal-agent settings is computationally tractable.
In particular, we show that the problem is solvable in polynomial time when either the number of agent's types $\ell$ or the number of outcomes $m$ is small.
Formally, we exhibit two algorithms that run in polynomial time when $\ell$ and $m$, respectively, are kept constant.

Let us remark that, as a byproduct of Theorem~\ref{thm:hard_label}, we also get that, even when the agent has a constant number of actions, it is \textsf{NP}-hard to approximate the contract-design problem up to within any given constant factor.
Specifically, the theorem implies that, for any $\rho \ge 1$, there exists a constant $k_\rho \in \mathbb{N}$ that depends on $\rho$ such that the problem is \textsf{NP}-hard to approximate up to within a multiplicative loss $\rho$ even when restricted to Bayesian principal-agent settings with $|A| \leq k_\rho$. 
However, the constant number of actions $k_\rho$ required for the hardness increases as the multiplicative approximation loss $\rho$ increases.
We leave as an open problem determining whether there are or not algorithms providing reasonable approximation guarantees with a small number of agent's actions.

\subsection{Constant Number of Types}

The crucial observation grounding our result is that the hardness of the problem of designing an optimal contract in Bayesian principal-agent settings stems from the difficulty of finding, among the exponentially-many possibilities, the tuple of agent's actions (one per type) that need to be incentivized.
Instead, given a tuple $\left( a_\theta \right)_{\theta \in \Theta}$ defining an agent's action $a_\theta \in A$ for each type $\theta \in \Theta$, a contract that implements $a_\theta$ for every type $\theta \in \Theta$ and maximizes the overall principal's expected utility can be obtained by the following linear program:
\begin{subequations}\label{lp:min_pay}
	\begin{align}
		\min_{p \in \mathbb{R}^m} & \quad  \sum_{\theta \in \Theta} \mu_\theta \sum_{\omega \in \Omega} p_\omega F_{\theta,a_\theta,\omega}\\
		\textnormal{s.t. } & \sum_{\omega \in \Omega} p_\omega F_{\theta,a_\theta,\omega} -c_{\theta,a_\theta} \ge \sum_{\omega \in \Omega} p_\omega F_{\theta,a,\omega} -c_{\theta,a} & \forall \theta \in \Theta, \forall a \in A \label{eq:incentive}\\
		&p_\omega \ge 0 & \forall \omega \in \Omega,
	\end{align} 
\end{subequations}
where, for the ease of notation, we identify a contract with a vector $p \in \mathbb{R}^m$ whose components are the payments $p_\omega$ for $\omega \in \Omega$ defining the contract.
Notice that, given that the agent's actions are fixed, the objective function to be minimized is the expected payment from the principal to the agent (as the principal's reward is fixed).
Constraints~\eqref{eq:incentive} ensure that each action $a_\theta$ is IC for an agent of type $\theta \in \Theta$ (recall that IR is ensured by Assumption~\ref{ass:ir}).

The following proposition shows that an optimal contract can be found by enumerating all the possible $n^\ell$ tuples of actions $\left( a_\theta \right)_{\theta \in \Theta}$, selecting the one that gives the highest optimal value for Problem~\eqref{lp:min_pay} (and the corresponding contract).
As an immediate consequence, we get that, when the number of agent's types $\ell$ is kept constant, then the overall running time of the resulting algorithm is polynomial in the size of the problem instance.~\footnote{The proofs of Theorem~\ref{prop:types} and Theorem~\ref{thm:outcomes} are deferred to the Appendix.}

\begin{restatable}{theorem}{propositionTypes}\label{prop:types}
	There exists an algorithm running in time polynomial in $n^\ell$ and $m$ that finds an optimal contract in any Bayesian principal-agent instance given as input.
\end{restatable}

\subsection{Constant Number of Outcomes}

The crucial insight underlying the polynomial-time algorithm is that, when the number of outcomes is kept constant, it is sufficient to search for an optimal contract in a polynomially-sized set of possible candidates.
For the ease of notation, we let $P \coloneqq \mathbb{R}_+^m$ be the set of vectors identifying all the possible contracts, where, given $p \in \mathbb{R}_+^m$, we denote with $p_\omega$ the vector component defining the payment associated to outcome $\omega \in \Omega$.
Moreover, for every agent's action $a \in A$ and agent's type $\theta \in \Theta$, we let $P(a,\theta) \subseteq P$ be the set identifying all the contracts that implement action $a$ for an agent of type $\theta$.
Formally, the set $P(a, \theta)$ is characterized by the following set of inequalities representing IC constraints:
\begin{equation}\label{eq:hyperplane}
	\sum_{\omega \in \Omega} p_\omega F_{\theta,a,\omega} -c_{\theta,a} \ge \sum_{\omega \in \Omega} p_\omega F_{\theta,a',\omega} -c_{\theta,a'} \quad \forall a' \in A : a' \neq a.
\end{equation}
Additionally, for every tuple of agent's actions $\textbf{a} = \left( a_\theta \right)_{\theta \in \Theta} \in \bigtimes_{\theta \in \Theta} A$, we let $P (\textbf{a}) \coloneqq \bigcap_{\theta \in \Theta} P(a_\theta, \theta)$ be the set identifying all the contracts that implement action $a_\theta$ for each agent's type $\theta \in \Theta$.
Finally, we let $P^\star \coloneqq \bigcup_{\textbf{a} \in \bigtimes_{\theta \in \Theta} A} \mathcal{V} (P(\textbf{a}))$, where $\mathcal{V}(P(\textbf{a}))$ denotes the set of vertices of polytope $P(\textbf{a})$.

The following theorem shows that, given any Bayesian principal-agent setting, there always exists an optimal contract belonging to the set $P^\star$ and that $P^\star$ has size bounded by a polynomial in $n^m$ and $\ell^m$.
Thus, whenever the number of outcomes $m$ is kept constant, an optimal contract can be computed in time polynomial in the size of the instance.

 \begin{restatable}{theorem}{theoremOutcomes}\label{thm:outcomes}
	There exists an algorithm running in time polynomial in $n^m$ and $\ell^m$ that finds an optimal contract in any Bayesian principal-agent instance given as input.
\end{restatable}

\section{Discussion}\label{sec:conclusions}

Despite principal-agent problems are ubiquitous in real-world economic scenarios, computational works on these problems appeared only recently and they are limited to specific settings~\citep{babaioff2012combinatorial,dutting2019simple,dutting2020complexity}.
In this paper, we introduce and study a new Bayesian principal-agent model in which the principal is uncertain about the agent's type.
This makes a considerable step over classical (non-Bayesian) principal-agent settings, as there are many real-world problems in which it is unreasonable to assume that the principal has complete knowledge of the agent.
Moreover, our Bayesian model begets new computational challenges that make it worth studying on its own, since, differently from the non-Bayesian case, in our setting a principal-optimal contract cannot be computed efficiently.

Linear contracts are the \emph{de facto} standard usually employed in real-world principal-agent problems, given their relative implementation simplicity, due to them being based on a pure-commission principle.
As a result, the research on principal-agent problems (mainly in economics, but also in computer science~\citep{dutting2019simple}) strived to find mathematical justifications of why linear contracts are so popular in practice.
Recently-developed studies show that, in non-Bayesian principal-agent settings, linear contracts are approximately optimal except in some degenerate situations~\citep{dutting2019simple} and that they enjoy some robustness properties~\citep{carroll2015robustness,carroll2019robustness,dutting2019simple}.
Our results further justify the use of linear contracts, showing that, in more realistic settings as those captured by our Bayesian model, they are the best among all the contracts that can be designed with bounded computationally resources.

\section*{Acknowledgments}
This work has been partially supported by the Italian MIUR PRIN 2017 Project ALGADIMAR ``Algorithms, Games, and Digital Market''.

\bibliographystyle{named}
\bibliography{biblio}

\appendix
\section{Proofs Omitted from Section~\ref{sec:tractable}}

\propositionTypes*

\begin{proof}
	The algorithm works by solving Problem~\eqref{lp:min_pay} for every possible tuple $\left( a_\theta \right)_{\theta \in \Theta}$ in the set $\bigtimes_{\theta \in \Theta} A$.
	Then, it picks the tuple (and the corresponding contract obtained by solving Problem~\eqref{lp:min_pay} for it) that results in the highest optimal value for Problem~\eqref{lp:min_pay}.
	We prove the correctness of the algorithm by showing that the returned contract, identified by a vector $ p^\star \in \mathbb{R}^m$, must provide the principal with an expected utility at least as large as that of any other contract.
	Let us take an arbitrary contract identified by vector $p \in \mathbb{R}^m$, and let $\left( a_\theta \right)_{\theta \in \Theta}$ be a tuple such that, for every $\theta \in \Theta$, the contract implements action $a_\theta$ for an agent of type $\theta$.
	Then, by solving Problem~\eqref{lp:min_pay} for $\left( a_\theta \right)_{\theta \in \Theta}$, the algorithm finds a contract incentivizing the same tuple of agent actions and requiring the principal an expected payment smaller than or equal to that of $p$.
	Notice that, since the contract found by solving the LP in Problem~\eqref{lp:min_pay} may lie on the boundary of its feasible region, there could be other tuples of agent actions that are incentivized by the contract.
	However, by using the assumption that the agent always breaks ties in favor of the principal, we can conclude that the tuple of agent actions that is actually played must provide the principal with an expected reward greater than or equal to that obtained for $\left( a_\theta \right)_{\theta \in \Theta}$.
	Thus, we can conclude that $p^\star$ provides the principal with an expected revenue greater than or equal to that of $p$, while requiring a smaller or equal payment, showing the correctness of the algorithm.
	Finally, notice that the algorithm solves $n^\ell$ different LPs, one for each tuple in $\bigtimes_{\theta \in \Theta} A$.
	The LPs have $m$ variables and $\ell \cdot n$ constraints, and, thus, they can be solved in time polynomial in $n$, $m$, and $\ell$.
\end{proof}

\theoremOutcomes*

\begin{proof}
	The proof involves two steps.
	\paragraph{First Step} We show that, for any contract defined by a vector $p \in P$, there exists another contract identified by a vector $p^\star \in P^\star$ providing the principal with an expected utility greater than or equal to that obtained for $p$.
	Let $\textbf{a}=(a_\theta)_{\theta \in \Theta} \in \bigtimes_{\theta \in \Theta} A$ be a tuple of agent actions such that the contract $p$ implements action $a_\theta$ for every type $\theta \in \Theta$.
	Let us define $p^\star \in P$ as the optimal solution of the LP in Problem~\eqref{lp:min_pay} for the tuple $(a_\theta)_{\theta \in \Theta}$.
	Noticing that the objective of Problem~\eqref{lp:min_pay} is to minimize a linear function over the polytope $P(\textbf{a})$, we can assume w.l.o.g. that the vector $p^\star$ is a vertex of the polytope, \emph{i.e.}, that $p^\star \in \mathcal{V} (P(\textbf{a}))$.
	Notice that, since $p^\star$ lies on a vertex of the feasible region of the LP, then there might be other tuples of agent actions that are incentivized by the contract identified by $p^\star$.
	However, given the assumption that the agent breaks ties in favor of the principal, these would provide the principal with an expected reward greater than or equal to that obtained for $(a_\theta)_{\theta \in \Theta}$.
	Thus, we can conclude that $p^\star$ has expected reward greater than or equal to that of $p$, while requiring a smaller or equal payment, proving the first step.
	
	\paragraph{Second Step}
	We show that the size of $P^\star$ can be bounded by a polynomial in $n^m$ and $\ell^m$.
	For any tuple of agent actions $\textbf{a}=(a_\theta)_{\theta \in \Theta} \in \bigtimes_{\theta \in \Theta} A$, the set $P(\textbf{a})$ is an $m$-dimensional polytope, and, thus, each vertex in $\mathcal{V}(P(\textbf{a}))$ is determined by the intersection of exactly $m$ hyperplanes among those defining it.
	Each polytopes $P(\textbf{a})$ is characterized by a subset of the hyperplanes defining the sets $P(a,\theta)$ for $a \in A$ and $\theta \in \Theta$.
	After removing duplicates, we can conclude that, for each $\theta \in \Theta$, there are at most $\binom{n}{2}$ hyperplanes resulting from Constraints~\ref{eq:hyperplane}, which are those defining the boundaries between the sets $P(a,\theta)$ and $P(a',\theta)$, for any pair of actions $a, a' \in A$ such that $a' \neq a$.
	Moreover, there are $m$ hyperplanes resulting from non-negativity constraints, namely $p_\omega \geq 0$ for every $\omega \in \Omega$.
	As a result, each polytope $P(\textbf{a})$ is defined by a subset of the same set of at most $\ell n^2 + m$ hyperplanes.
	Hence, each vertex in $P^\star$ is obtained as the intersection of exactly $m$ of these at most $\ell n^2 + m$ hyperplanes and we can conclude that there are at most $\binom{\ell n^2+m}{m}$ vertices in $P^\star$.
	%
	In conclusion, to find an optimal contract it is sufficient that the algorithm enumerates all the vertices in $P^\star$, which requires time polynomial in $n^m$ and $\ell^m$.
\end{proof}

\end{document}